\documentclass[pra,twocolumn,floatfix,superscriptaddress,longbibliography,notitlepage]{revtex4-1}

\usepackage{graphicx, color, graphpap}

\usepackage{enumitem}
\usepackage{amssymb}
\usepackage{amsthm}
\usepackage{multirow}
\usepackage[colorlinks=true,citecolor=blue,linkcolor=magenta]{hyperref}
\usepackage[T1]{fontenc}

\usepackage{bbm}
\usepackage{thmtools,thm-restate}
\usepackage{verbatim}
\usepackage{mathtools}
\usepackage{titlesec}
\usepackage{amsmath}
\usepackage[title]{appendix}
\usepackage[linesnumbered,ruled,vlined]{algorithm2e}
\SetKwInput{kwInit}{Init}

\usepackage[caption = false]{subfig}

\long\def\ca#1\cb{} 

\newcommand{\ket}[1]{|#1\rangle}               
\newcommand{\bra}[1]{\langle #1|}              
\newcommand{\dya}[1]{\ket{#1}\!\bra{#1}}

\newcommand{\poly}{\operatorname{poly}}

\newcommand{\EC}{\mathcal{E}}

\newcommand{\HC}{\mathcal{H}}

\newcommand{\LC}{\mathcal{L}}

\newcommand{\OC}{\mathcal{O}}

\newcommand{\RC}{\mathcal{R}}
\newcommand{\SC}{\mathcal{S}}

\newcommand{\Tr}{{\rm Tr}}

\newcommand{\Var}{{\rm Var}}
\newcommand{\Cov}{{\rm Cov}}

\renewcommand{\geq}{\geqslant}
\renewcommand{\leq}{\leqslant}

\DeclareMathOperator*{\argmin}{arg\,min}
\renewcommand{\vec}[1]{\boldsymbol{#1}}  

\newcommand{\ad}{^\dagger}

\newcommand*{\id}{\openone}

\newcommand{\thv}{\vec{\theta}}

{}
{}
\newtheorem{theorem}{Theorem}

\newtheorem{corollary}{Corollary}

\newtheorem{proposition}{Proposition}

\newtheorem{suplemma}{Supplemental Lemma}

\newtheorem{remark}{Remark}

\usepackage{amssymb}
\usepackage{dsfont}

\begin{document}

\title{Subtleties in the trainability of quantum machine learning models}

\author{Supanut Thanaslip}
\affiliation{Theoretical Division, Los Alamos National Laboratory, Los Alamos, New Mexico 87545, USA}
\affiliation{Centre for Quantum Technologies, National University of Singapore, 3 Science Drive 2 117543, Singapore}

\author{Samson Wang}
\affiliation{Theoretical Division, Los Alamos National Laboratory, Los Alamos, New Mexico 87545, USA}
\affiliation{Imperial College London, London, UK}

\author{Nhat A. Nghiem}
\affiliation{Theoretical Division, Los Alamos National Laboratory, Los Alamos, New Mexico 87545, USA}
\affiliation{Department of Physics and Astronomy, State University of New York at Stony Brook, Stony Brook, New York 11794-3800, USA}

\author{Patrick J. Coles}
\affiliation{Theoretical Division, Los Alamos National Laboratory, Los Alamos, New Mexico 87545, USA}

\author{M. Cerezo}
\affiliation{Information Sciences, Los Alamos National Laboratory, Los Alamos, NM 87545, USA}
\affiliation{Center for Nonlinear Studies, Los Alamos National Laboratory, Los Alamos, New Mexico 87545, USA}

\begin{abstract}
A new paradigm for data science has emerged, with quantum data, quantum models, and quantum computational devices. This field, called Quantum Machine Learning (QML), aims to achieve a speedup over traditional machine learning for data analysis. However, its success usually hinges on efficiently training the parameters in quantum neural networks, and the field of QML is still lacking theoretical scaling results for their trainability. Some trainability results have been proven for a closely related field called Variational Quantum Algorithms (VQAs).  While both fields involve training a parametrized quantum circuit, there are crucial differences that make the results for one setting not readily applicable to the other. In this work we bridge the two frameworks  and  show that gradient scaling results for VQAs can also be applied to study the gradient scaling of QML models. Our results indicate that features deemed detrimental for VQA trainability can also lead to issues such as barren plateaus in QML. Consequently, our work has implications for several QML proposals in the literature. In addition, we provide theoretical and numerical evidence that QML models exhibit further trainability issues not present in VQAs, arising from the use of a training dataset. We refer to these as dataset-induced barren plateaus.  These results are most relevant when dealing with classical data, as here the choice of embedding scheme (i.e., the map between classical data and quantum states) can greatly affect the gradient scaling.
\end{abstract}

\maketitle

\section{Introduction}

The future of data analysis is incredibly exciting. The quantum revolution promises new kinds of data, new kinds of models, and new information processing devices. This is all made possible because small-scale quantum computers are currently available, while larger-scale ones are anticipated in the future~\cite{preskill2018quantum}. The mere fact that users will run jobs on these devices, preparing interesting quantum states, implies that new datasets will be generated. These are called quantum datasets~\cite{perrier2021qdataset,schatzki2021entangled}, as they exist on the quantum computer and hence must be analyzed on the quantum computer. Naturally this has led to the proposal of new models, so-called Quantum Neural Networks (QNNs)~\cite{schuld2014quest}, for analyzing (classifying, clustering, etc.) such data. Different architectures have been proposed for these models: dissipative QNNs~\cite{beer2020training}, convolutional QNNs~\cite{cong2019quantum}, recurrent QNNs~\cite{bausch2020recurrent}, and others~\cite{farhi2018classification,verdon2019quantum}.

Using quantum computers for data analysis is often called Quantum Machine Learning (QML)~\cite{biamonte2017quantum,schuld2015introduction}. This paradigm is general enough to also allow for analysis of classical data. One simply needs an embedding map that first maps the classical data to quantum states, and then such states can be analyzed by the QNN~\cite{havlivcek2019supervised,larose2020robust,lloyd2020quantum}. Here, the hope is that by accessing the exponentially large dimension of the Hilbert space and quantum effects like superposition and entanglement, QNNs can outperform their classical counterparts (i.e., neural networks) and achieve a coveted quantum advantage~\cite{huang2021information,huang2021power,kubler2021inductive,liu2021rigorous,aharonov2021quantum}.

Despite the tremendous promise of QML, the field is still in its infancy and rigorous results are  needed to guarantee its success. Similar to classical machine learning, here one also wishes to achieve small training error~\cite{larocca2021theory} and small generalization error~\cite{banchi2021generalization,du2021efficient}, with the second usually hinging on the first. Thus, it is crucial to study how certain properties of a QML model can hinder or benefit its parameter trainability.

Such trainability analysis has become a staple in a closely related field called Variational Quantum Algorithms (VQAs)~\cite{cerezo2020variationalreview}. This field aims to optimize over sets of quantum circuits to discover efficient versions of quantum algorithms for accomplishing various tasks, such as finding ground states~\cite{peruzzo2014variational}, solving linear systems of equations~\cite{bravo2020variational,huang2019near,xu2019variational}, quantum compiling~\cite{khatri2019quantum,sharma2019noise}, factoring~\cite{anschuetz2019variational}, and dynamical simulation~\cite{cirstoiu2020variational,commeau2020variational,endo2020variational,li2017efficient}. In this field, a great deal of effort has been put forward towards analyzing and avoiding the barren plateau  phenomenon~\cite{mcclean2018barren,cerezo2020cost,larocca2021diagnosing,marrero2020entanglement,patti2020entanglement,holmes2021connecting,holmes2021barren,huembeli2021characterizing,zhao2021analyzing,wang2020noise}.  When a VQA exhibits a barren plateau, the cost function gradients vanish  exponentially with the problem size, leading to an exponentially flat optimization landscape.  Barren plateaus greatly impact the trainability of the parameters as it is impossible to navigate  the flat landscape without expending an exponential amount of resources~\cite{cerezo2020impact,arrasmith2020effect,wang2021can}.

Given the large body of literature studying barren plateaus in VQAs, the natural question that arises is: Are the gradient scaling and barren plateau results also applicable to QML? Making this connection is crucial, as it is not uncommon for QML proposals to employ certain features that have been shown to be detrimental for the trainability of VQAs (such as deep unstructured circuits~\cite{mcclean2018barren,cerezo2020cost} or global measurements~\cite{cerezo2020cost}). Unfortunately, it is not straightforward to directly employ VQA gradient scaling results in QML models. Moreover, one can expect that in addition to the trainability issues arising in variational algorithms, other problems can potentially appear in QML schemes. This is due to the fact that QML models are generally more complex. For instance, in QML one needs to deal with datasets~\cite{schatzki2021entangled}, which further require the use of an embedding scheme when dealing with classical data~\cite{havlivcek2019supervised,lloyd2020quantum,schuld2020circuit,larose2020robust}.  Additionally,  QML loss functions can be more complicated than  VQA cost functions, as the latter are usually linear functions of expectation values of some set of operators.

In this work we study the trainability and the existence of barren plateaus in QML models. Our work represents a general treatment that goes beyond previous analysis of gradient scaling and trainability in specific QML models~\cite{pesah2020absence,sharma2020trainability,liu2021presence,abbas2020power,haug2021optimal,kieferova2021quantum,kiani2021quantum,tangpanitanon2020expressibility}. Our main results are two-fold. First, we rigorously connect the scaling of gradients in VQA-type cost functions and QML loss functions, so that barren plateau results in the variational algorithms literature can be used to study the trainability of QML models. This implies that known sources of barren plateaus extend to the realm of training QNNs, and thus, that many proposals in the literature need to be revised. Second, we present theoretical and numerical evidence  that certain features in the datasets and embedding can additionally greatly affect the model's trainability. These results show that additional care must be taken when studying the trainability of QML models. Moreover, they constitute a novel source for untrainability: dataset-induced barren plateaus.

\section{Framework}

\subsection{Quantum Machine Learning}\label{sec:framework-QML}

In this work, we consider supervised Quantum Machine Learning (QML) tasks. First, as depicted in Fig.~\ref{fig:1}(a) let us consider the case where one has classical data. Here, the dataset is of the form $\{\vec{x}_i,y_i\}$, where $\vec{x}_i\in X$ is some classical  data  (e.g., a real-valued vector), and $y_i\in Y$ are values or labels associated with each $\vec{x}_i$ according to some (unknown) model $h:X\rightarrow Y$.   One generates a training set $\SC=\{\vec{x}_i,y_i\}^{N}_{i=1}$ of size $N$  by sampling from the dataset according to some probability distribution. Using $\SC$ one trains a QML model, i.e. a parametrized map $\tilde{h}_{\thv}:X\rightarrow Y$, such that its predictions agree with those of $h$ with high probability on $\SC$ (low training error), and on previously unseen cases (low generalization error).

For the QML model to access the exponentially large dimension of the Hilbert space, one needs to encode the classical data into a quantum state. As shown in Fig.~\ref{fig:1}(a), one initializes $m$ qubits in a fiduciary state such as $\ket{\vec{0}}=\ket{0}^{\otimes m}$ and sends it through a Completely Positive Trace-Preserving (CPTP) map $\EC_{\vec{x}_i}^E$, so that the outputs are of the form 
\begin{equation}\label{eq:embedding-output}
    \rho_i=\EC_{\vec{x}_i}^E(\dya{\vec{0}})\,.
\end{equation}
For instance, $\EC_{\vec{x}_i}^E$ can be a unitary channel given by the action of a parametrized quantum circuit whose gate rotation angles depends on the values in $\vec{x}_i$. While in general the embedding can be in itself trainable~\cite{lloyd2020quantum, hubregtsen2021training, thumwanit2021trainable}, here the embedding is fixed. 

Next, consider the case when the dataset is composed of quantum data~\cite{cong2019quantum,schatzki2021entangled}. As one can see from Fig.~\ref{fig:1}(b), here the dataset is of the form $\{\rho_i,y_i\}$, where $\rho_i\in R$, and where  $R\subseteq \HC$ is a set of quantum states in the Hilbert space $\HC$. Then, $y_i\in Y$ are values or labels associated with each quantum state according to some (unknown) model $h:R\rightarrow Y$. In what follows we simply denote as $\rho_i$ the states obtained from the dataset, and clarify when appropriate if they were produced using an embedding scheme or not.

\begin{figure}[t!]
	\includegraphics[width= 1 \columnwidth]{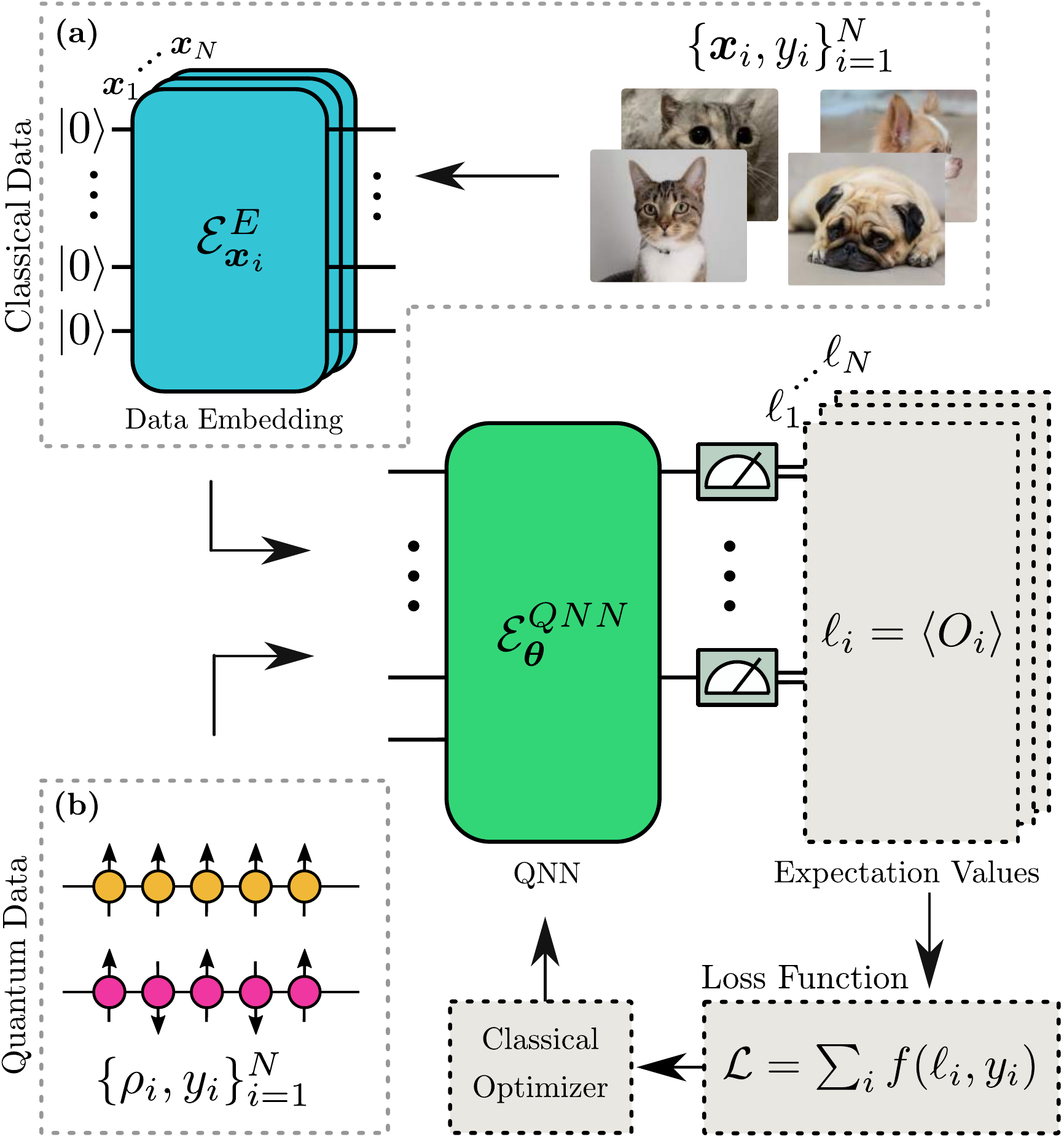}
	\caption{\textbf{Schematic diagram of a QML task.} Consider a supervised learning task where the training dataset is either (a) classical  or (b) quantum. The classical data points are of the form $\{\vec{x}_i,y_i\}$, where $\vec{x}_i\in X$ are input data  (e.g. pictures) and  $y_i\in Y$ are labels associated with each $\vec{x}_i$ (e.g., cat/dog). An embedding channel $\EC_{\vec{x}_i}^E$ maps the classical data onto quantum states $\rho_i$. Alternatively, quantum data points are of the form $\{\rho_i,y_i\}$, where $\rho_i$ are  quantum states in a Hilbert space $\HC$, each with associated labels $y_i\in Y$ (e.g., ferromagnetic/paramagnetic phases). The quantum states (coming from classical or quantum datasets) are then sent through a parametrized quantum neural network (QNN), $\EC_{\thv}^{QNN}$. By performing measurements on the output states one estimates expectation values $\ell_i\equiv \ell_i(\thv;y_i)$ which are then used to estimate the loss function $\LC(\thv)$. Finally, a classical optimizer decides how to adjust the QNN parameters, and the loop repeats multiple times to minimize the loss function.   }
	\label{fig:1}
\end{figure}

The quantum states $\rho_i$ are then sent through a QNN, i.e., a parametrized CPTP map $\EC_{\thv}^{QNN}$. Here, the outputs of the QNN are $n$-qubit quantum states (with $n\leq m$) of the form 
\begin{equation}\label{eq:qnn-output}
    \rho_i(\thv)=\EC_{\thv}^{QNN}(\rho_i)\,.
\end{equation}
We note that Eq.~\eqref{eq:qnn-output} encompasses most widely used QNN architectures. For instance, in many cases  $\EC_{\thv}^{QNN}$ is simply a trainable parametrized quantum circuit~\cite{farhi2018classification,cong2019quantum,beer2020training,bausch2020recurrent}. 
Here, $\thv$ is a vector of continuous variables (e.g., gate rotation angles). More generally, $\thv$ could also contain discrete variables (gate placements)~\cite{grimsley2019adaptive,du2020quantum,bilkis2021semi}. However, for simplicity, here we consider the case when one only optimizes over continuous QNN parameters.

The model predictions are obtained by measuring qubits of the QNN output states $\rho_i(\thv)$ in Eq.~\eqref{eq:qnn-output}. That is, for a given data point $(\rho_i,y_i)$, one estimates the quantity   
\begin{align}\label{eq:measurement}
\begin{split}
    \ell_i(\thv;y_i) =\Tr[\rho_i(\thv)O_{y_i}]\,,
\end{split}
\end{align}
where $O_{y_i}$ is a Hermitian operator. Here we recall that we use the term global measurement when $O_{y_i}$ acts non-trivially on all $n$ qubits. On the other hand, we say that the measurement is local when $O_{y_i}$ acts non-trivially on a small constant number of qubits (e.g. one or two qubits). 
Finally, throughout training, the success of the QML model is quantified via a loss function $\LC(\thv)$ that takes the form
\begin{align}\label{eq:generic-loss-function}
    \LC(\thv)& = \sum_{i=1}^N f(\ell_i(\thv;y_i),y_i)\,,
\end{align}
where $f(.)$ is first-order differentiable. By employing a classical optimizer, one trains the parameters in the QNN  to solve the optimization task
\begin{equation}
    \thv_*=\argmin_{\thv} \LC(\thv)\,.
\end{equation}
Finally,  one tests the generalization performance by assigning labels with the optimized model $h_{\vec{\thv}_*}$.

For the purpose of clarity, let us exemplify the previous concepts. Consider a binary classification task where the labels are $Y=\{-1,1\}$. Here, one possible option to assign label $1(-1)$ is to measure globally all qubits in the computational basis and estimate the number of output bitstrings with even (odd) parity. That is, the probability of assigning label $y_i$ is given by the expectation value
\begin{equation}\label{eq:prob-label}
     p_i(y_i|\thv)=\Tr[\rho_i(\thv)O_{y_i}]\,,
\end{equation}
where $O_{1}=\sum_{\vec{z}:even}\dya{\vec{z}}$ and $O_{-1}=\sum_{\vec{z}:odd}\dya{\vec{z}}$. Alternatively, rather than computing the probability of assigning a given label, the measurement outcome can be a label prediction by itself. For instance, the latter can be achieved with the global measurement
 \begin{align}\label{eq:label-pred}
     \tilde{y}_i(\thv)=\Tr[\rho_i(\thv)Z^{\otimes n}]\,,
 \end{align}
as this is a number in $[-1,1]$.

Here, let us make several important remarks. First, we note that Eq.~\eqref{eq:prob-label} and~\eqref{eq:label-pred} are precisely of the form in Eq.~\eqref{eq:measurement}, where one computes the expectation value of an operator over the output state of the QNN. Then, we remark that both approaches are equivalent due to the fact that $O_{1}-O_{-1}=Z^{\otimes n}$, and thus
\begin{equation}
    \tilde{y}_i(\thv)= p_i(1|\thv)-p_i(-1|\thv)\,,
\end{equation}
and conversely
\begin{equation}\label{eq:prob-parity}
    p_i(y_i|\thv)=\frac{1+y_i\tilde{y}_i(\thv)}{2}\,.
\end{equation}
Finally, for the purpose of accuracy testing one needs to map from the continuous set of outcomes of $\tilde{y}_i(\thv)$ or $p_i(y_i|\thv)$ to the discrete set $Y$. Thus,  the QML model $\tilde{h}_{\thv}$  assigns label  $y_i$ if $p_i(y_i|\thv)\geq  p_i(-y_i|\thv)$.

Here it is also worth recalling two widely used loss functions in the literature. First, the mean squared error loss function is given by
\begin{align}
    \LC_{mse}(\thv) = \frac{1}{N}\sum_{i=1}^{N}(\tilde{y}_i(\thv)-y_i)^2\,, \label{eq:loss-mse}
\end{align}
where $\tilde{y}_i(\thv)$ is the model-predicted label for the data point $\rho_i$ obtained through some expectation value as in Eq.~\eqref{eq:measurement} (see for instance the label of Eq.~\eqref{eq:label-pred}). Then, the negative log-likelihood loss function  is defined as
\begin{align}
    \LC_{log}(\thv) &=-\frac{1}{N} \log \left(\prod_{i=1}^{N} p_i(y_i|\thv)\right)\nonumber \\ 
    &= - \frac{1}{N} \sum_{i = 1}^{N} \log (p_i(y_i|\thv))\,, \label{eq:loss-log}
\end{align}
where $p_i(y_i|\thv)$ is the probability of assigning label $y_i$ to the data point $\rho_i$ obtained through some expectation value as in Eq.~\eqref{eq:measurement} (see for instance the probability of Eq.~\eqref{eq:prob-label}). Moreover, we recall that the expectation of the negative log-likelihood Hessian is given by the Fisher Information (FI) matrix. In practice, one can estimate the FI matrix via the empirical FI matrix
\begin{align}\label{eq:FI_empirical}
    \tilde{F}(\vec \theta) = \frac{1}{N} \sum_{i=1}^{N} \frac{\partial}{\partial \vec \theta} \log p_i(y_i|\thv) \frac{\partial}{\partial \vec \theta} \log p_i(y_i|\thv)^T.
\end{align}

The FI matrix plays a crucial role in natural gradient optimization methods~\cite{amari1998natural}. Here, the FI measures the sensitivity of the output probability distributions to parameter updates. Hence, such optimization methods rely on the estimation of the FI matrix to optimize the QNN parameters.  Below we discuss how such optimization methods are affected when the QML model exhibits a vanishing gradient.

\subsection{Quantum Landscape Theory and Barren Plateaus}

Recently, there has been a tremendous effort in developing the so-called Quantum Landscape Theory~\cite{arrasmith2021equivalence,larocca2021theory} to analyze  properties of quantum loss landscapes, how they emerge, and how they affect the parameter optimization process. Here, one of the main topics of interest is that of Barren Plateaus (BPs). When a loss function exhibits a BP, the gradients vanish exponentially with the number of qubits, and thus an exponentially large number of shots are needed to navigate through the flat landscape. 

Since BPs have been mainly studied in the context of Variational Quantum Algorithms (VQAs)~\cite{cerezo2020variationalreview,mcclean2018barren,cerezo2020cost,larocca2021diagnosing,marrero2020entanglement,patti2020entanglement,holmes2021connecting,holmes2021barren,huembeli2021characterizing,zhao2021analyzing,wang2020noise}, we here briefly recall that in a VQA implementation, the goal is to minimize a linear cost function that is usually of the form 
\begin{equation}\label{eq:cost-function}
    C(\thv)=\Tr[U(\thv)\rho U\ad(\thv) H]\,.
\end{equation}
Here $\rho$ is the initial state, $U(\thv)$ a trainable parametrized quantum circuit, and $H$ a Hermitian operator. BPs for cost functions such as that in Eq.~\eqref{eq:cost-function} have been shown to arise for global operators $H$~\cite{cerezo2020cost}, as well as due to certain properties of $U(\thv)$ such as its structure and depth~\cite{mcclean2018barren,larocca2021diagnosing,wang2020noise}, its expressibility~\cite{holmes2021connecting,holmes2021barren}, and its entangling power~\cite{sharma2020trainability,marrero2020entanglement,patti2020entanglement}. 

There are several ways in which  the BP phenomenon can be characterized. The most common is through  the scaling of the variance of cost function partial derivatives, as in a BP one has 
\begin{equation}\label{eq:var-BP}
    \Var[\partial_\nu C(\thv)] \in \OC(1/\alpha^n)\,,
\end{equation}
where $\partial_\nu C(\thv)=\partial C(\thv)/\partial \theta_\nu$, 
$\theta_\nu\in\thv$, and for  $\alpha > 1$. Here, the variance is taken with respect to the set of parameters $\thv$. We recall that here one also has that $\mathbb{E}[\partial_\nu C(\thv)]=0$, implying that gradients exponentially concentrate at zero. In addition, a  BP can also be characterized through the concentration of cost function values, so that $\Var[ C(\thv_A)-C(\thv_B)] \in \OC(1/\beta^n)$, where $\thv_A$ and $\thv_B$ are two randomly sampled sets of parameters, and $\beta>1$~\cite{arrasmith2021equivalence}. Finally, the presence of a BP can be diagnosed by inspecting the eigenvalues of the Hessian and the FI matrix, as these become exponentially vanishing with the system size~\cite{huembeli2021characterizing,abbas2020power} (see also Section~\ref{sec:fi-matrix} below). 

Despite tremendous progress in understanding BPs for VQAs, there are only a few results which analyze the existence of BPs in QML settings~\cite{pesah2020absence,sharma2020trainability,liu2021presence,abbas2020power,haug2021optimal,kieferova2021quantum,kiani2021quantum}. In fact, while in the VQA community there is a consensus that certain properties of an algorithm should be avoided (such as global measurements), these same properties are still widely used in the QML literature (e.g., global parity measurements as those in Eqs.~\eqref{eq:prob-label} and~\eqref{eq:label-pred}).   Thus, bridging the VQA and QML communities would allow the use of trainability and gradient scaling results across both fields.

\subsection{VQAs versus QML}

Let us first recall that, as previously mentioned, training a VQA or the QNN in a QML model usually implies optimizing the parameters in a quantum circuit. However, despite this similarity, there are some key differences between the VQA and a QML framework, which we summarize in Table~\ref{table:VQAvsQML}. First,  QML is based on training a QNN using data from dataset. Moreover, the input quantum stats to the QML model contain information that the model is trying to learn. On the other hand, in VQAs there is usually no training data, but rather a single input state sent trough a parametrized quantum circuit. Here, such initial state is generally an easy-to-prepare state such as $\ket{\vec{0}}$.

Then, while VQAs generally deal with quantum states, QML models are envisioned to work both on classical and quantum data. Thus, when using a QML model with classical data, there is an additional encoding step that depends on the choice of embedding channel $\EC_{\vec{x}_i}^E$. When the QML model deals with quantum data, the input states are usually non-trivial states.

Finally, we note that for most VQAs, the cost function of Eq.~\eqref{eq:cost-function} is a linear function in the matrix elements of $U(\thv)\rho U\ad(\thv)$. For QML, however, the loss function of Eq.~\eqref{eq:generic-loss-function} can be a more complex non-linear function of the matrix elements of $\rho_i(\thv)$ (see for instance the log-likelihood in Eq.~\eqref{eq:loss-log}). Here, it is still worth noting that the expectation values $\ell_i(\thv;y_i)$ of Eq.~\eqref{eq:measurement} are exactly of the form of VQA cost functions, as both of these are  linear functions of the parametrized output states.

\begin{table}[t]
\centering
\begin{tabular}{||c | c | c||}
    \hline
      & VQA  & QML\\
    \hline \hline
    Input & Single (trivial)  & Dataset of   \\
            & input state     &  non-trivial states\\
    \hline
    Embedding & Not required  & Required for \\
     & &  classical data \\
    \hline
    Loss function & Linear & Non-linear\\
    \hline
\end{tabular}
\caption{\textbf{Differences between VQAs and QML}. This table summarizes the key differences between the VQA and QML settings, including the input that the algorithm takes, whether an embedding scheme is needed, and the form of the loss (or cost) function. }
\label{table:VQAvsQML}
\end{table}

\section{Analytical Results}\label{section:analytic-results}

The differences between the VQA and QML frameworks  discussed in the previous section mean that known gradient scaling results for VQAs do not automatically or trivially extend to the QML setting. Nevertheless, the goal of this section is to establish and formalize the connection between these two frameworks, as summarized in Fig.~\ref{fig:summary-results}. Namely, in Section~\ref{sec:connecting-loss-cost} we link the variance of the partial derivative of QML loss functions of the form of Eq.~\eqref{eq:generic-loss-function} to that of partial derivatives of linear expectation values, which are the quantities of study for VQAs (see Eq.~\eqref{eq:cost-function}). This allows us to show that, under standard assumptions, the landscapes of QML loss functions such as the mean-squared error loss in Eq.~\eqref{eq:loss-mse} and the negative log-likelihood loss in Eq.~\eqref{eq:loss-log} have BPs in all the settings that standard VQAs do. In addition, as we discuss in Section~\ref{sec:embedding-induced-BP}, additional care should be taken to guarantee the trainability of a QML model, as here one has to deal with datasets and embeddings. This leads us to introduce a new mechanism for untrainability arising from the dataset itself. Finally, in Section~\ref{sec:fi-matrix} we demonstrate that under conditions which induce a BP for the negative log-likelihood loss function, the matrix elements of the empirical Fisher Information matrix also exponentially vanish.

\begin{figure}[t!]
	\includegraphics[width=\columnwidth]{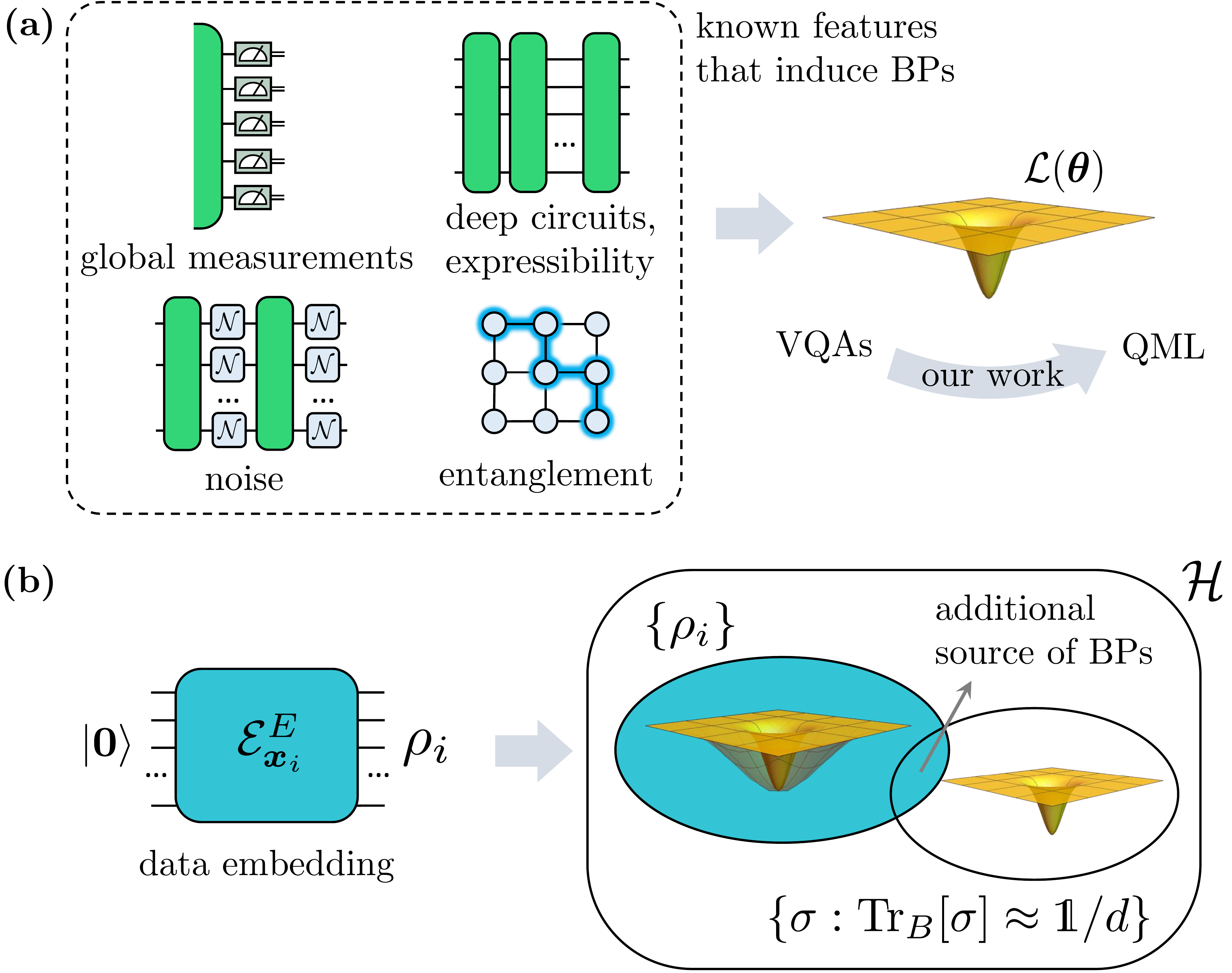}
	\caption{\textbf{Summary of results.} (a) In this work we bridge the gap between trainability results for the VQA and QML settings. Namely, we show that the considered class of supervised QML models will exhibit a barren plateau (BP) in all settings that standard VQAs do. This means that features such as global measurements, deep circuits, or highly entangling circuits should be avoided in QML. (b) We present analytical and numerical evidence that aspects of the dataset can be an additional source of BPs. For instance,  embedding schemes for classical data can lead to states that have highly mixed reduced states and thus display trainability issues. This is a novel source for BPs which we call dataset-induced BPs.}
	\label{fig:summary-results}
\end{figure}

\subsection{Connecting the gradient scaling of QML loss functions and linear cost functions}\label{sec:connecting-loss-cost}

In the following theorem we study the variance of the partial derivative for the generic loss function defined in Eq.~\eqref{eq:generic-loss-function}. 
We upper bound this quantity by making a connection to the variance of partial derivatives of linear expectation values of the form of Eq.~\eqref{eq:measurement}. As shown in Section \ref{sec:appdx-theorem} of the Appendix, the following theorem holds.

\begin{theorem}[Variance of partial derivative of generic loss function]\label{thm:var-generic-loss}
Consider the partial derivative of the loss function $\LC(\thv)$ in Eq.~\eqref{eq:generic-loss-function} taken with respect to variational parameter $\theta_\nu \in\thv$. We denote this quantity as $\partial_\nu \LC(\thv) \equiv \partial \LC(\thv)/\partial\theta_\nu$. The following inequality holds \small
\begin{align} \label{eq:supp-var-generic-loss}
    \Var[\partial_\nu \LC(\thv)] \leq \left(\frac{1}{N} \sum_i {g_i} \sqrt{ \Var[\partial_\nu \ell_i(\thv)]+ (\mathbb{E}[\partial_\nu \ell_i(\thv)])^2} \right)^2,
\end{align}
\normalsize
where we used $\ell_i(\thv)$ as a shorthand notation for $\ell_i(\thv;y_i)$ in Eq.~\eqref{eq:measurement}, and where the expectation values are taken over the parameters $\thv$. Here we defined
\begin{align}
    g_i = \sqrt{2}\max_{\ell_i}\left|\frac{\partial f}{\partial \ell_i}\right|\,, \label{eq:g_i}
\end{align}
where $\frac{\partial f}{\partial \ell_i}$ is the $i$-th entry of the Jacobian $J_{\vec{\ell}}$ with $\vec{\ell}=(\ell_1,\ldots,\ell_N)$, and where we denote as $\max_{\ell_i}\left|\frac{\partial f}{\partial \ell_i}\right|$ the maximum value of the partial derivative of $f(\ell_i,y_i)$.
\end{theorem}
Theorem~\ref{thm:var-generic-loss} establishes a formal relationship between gradients of generic loss functions and linear expectation values, which includes linear cost functions in VQA frameworks. This result provides a bridge to bound the variance of partial derivatives of QML loss functions based on known behavior of linear cost functions. As shown next, this allows us to use gradient scaling results from the VQA literature in QML settings.

The direct implication of Theorem \ref{thm:var-generic-loss} can be understood as follows. Consider the case when $g_i$ is at most polynomially increasing in $n$ for all $\rho_i$ in the training set. Then, if a facet of the model causes linear cost functions to exhibit BPs, then $\LC(\thv)$ will  also suffer from BPs. That is, if $\Var[\partial_\nu \ell_i(\thv)]$ is exponentially vanishing, then so will be $\Var[\partial_\nu \LC(\thv)]$.  Let us now explicitly demonstrate the implications of Theorem \ref{thm:var-generic-loss} for the mean squared error and the negative log-likelihood cost functions.

\begin{corollary}[Barren plateaus in mean squared error and negative log-likelihood loss functions]\label{corollary:cor1}
Consider the mean squared error loss function $\LC_{mse}$ defined in Eq.~\eqref{eq:loss-mse} and the negative log-likelihood loss $\LC_{log}$ defined in Eq.~\eqref{eq:loss-log}, with respective model-predicted labels $\tilde{y}_i(\thv)$ and model-predicted probabilities $p_i(y_i|\thv)$ both of the form of Eq.~\eqref{eq:measurement}. Suppose that the QNN has a BP for the linear expectation values, i.e., Eq.~\eqref{eq:var-BP} is satisfied for all $\tilde{y}_i(\thv)$ and $p_i(y_i|\thv)$. Then, assuming that $\tilde{y}_i(\thv) \in \OC(\poly(n))\; \forall i$ we have 
\begin{align}
    \Var[\partial_\nu \LC_{mse}] \in \mathcal{O}(1/\alpha^{n})\, ,
\end{align}
for  $\alpha > 1$. Similarly, assuming that $p_i(y_i|\thv) \in [b,1]\; \forall i,\thv$, where $b\in \Omega(1/\poly(n))$, we have
\begin{align}
    \Var[\partial_\nu \LC_{log}] \in \mathcal{O}(1/\alpha^{n})\, ,
\end{align}
for  $\alpha > 1$. 
\end{corollary}

We note the assumption on $\tilde{y}_i(\thv)$ in Corollary~\ref{corollary:cor1} can be made to be generically satisfied as label values are usually bounded by a constant, and if not they can always be normalized by classical post-processing. The assumption on the possible values of $p_i(y_i|\thv)$ is equivalent to clipping model-predicted probabilities such that they are strictly greater than zero, which is a common practice when using the negative log-likelihood cost function~\cite{pedregosa2011scikit, abadi2016tensorflow}.

Corollary~\ref{corollary:cor1} explicitly implies that, under mild assumptions, previously established BP results for VQAs are applicable to QML models that utilize mean square error and negative log-likelihood loss functions. Consequently, existing QML proposals with these BP-induced features, such those as summarized in Fig.~\ref{fig:summary-results}, need to be revised in order to avoid an exponentially flat loss function landscape.

We note that, as shown in the Appendix~\ref{sec:appdx-cor-extensions}, the above results in Corollary~\ref{corollary:cor1} can also be extended to the generalized mean square error loss function used in multivariate regression tasks and to loss functions based on the Kullback–Leibler divergence used in generative modelling.

\subsection{Dataset-induced barren plateaus}\label{sec:embedding-induced-BP}

In this section we argue that the dataset can negatively impact QML trainability if the input states to the QNN have high levels of entanglement and if the QNN employs local gates (which is the standard assumption in most widely used QNNs~\cite{farhi2018classification,cong2019quantum,beer2020training,bausch2020recurrent}). This is due to the fact that reduced states of highly entangled states can be very close to being maximally mixed, and it is hard to train local gates acting on such states. Note that this phenomenon is not expected to arise in a standard VQA setting where the input state is considered to be a trivial tensor-product state.

To illustrate this issue, we will present an example where a VQA does not exhibit a BP, but a QML model can still have a dataset-induced BP. Consider a model where the QNN is given by a single layer of the so-called alternating layered ansatz~\cite{cerezo2020cost}. Here, $\EC_{\thv}^{QNN}$ is a circuit composed a tensor product of $\xi$ $s$-qubit unitaries  such that $ n = s \xi$. That is, the QNN is of the form $U(\thv)=\bigotimes_{k=1}^{\xi} U_k(\thv_k)$,  and where $U_k(\thv_k)$ is a general unitary acting on $s$ qubits.   Furthermore, consider a linear loss function constructed from local expectation values of the form
\begin{align}
   \ell_i(\thv) = 1 - \frac{1}{n} \sum_{j=1}^n {\rm Tr} \left[ \left(\ket{0}\bra{0}_j \otimes \id_{\bar j}\right)  \rho_i(\thv)  \right]\,, \label{eq:cost-local}
\end{align}
where $\dya{0}_j$ denotes the projector on the $j$-th qubit, and where $\rho_i(\thv)$ is defined in Eq.~\eqref{eq:qnn-output}. 
We now quote the following result from Ref.~\cite{cerezo2020cost} on the variance of the partial derivative of quantities of the form in Eq.~\eqref{eq:cost-local}. 

\begin{proposition}[From Supplementary Note 5 of Ref.~\cite{cerezo2020cost}]\label{prop:local-cost}
Suppose that $\EC_{\thv}^{QNN}$ is given by an application of a single layer of the alternating layered ansatz consisting of a tensor product of $s$-qubit unitaries.  Consider the partial derivative of the local cost in Eq.~\eqref{eq:cost-local} taken with respect to a parameter $\theta_\nu$ appearing in the $h$-th unitary. If each unitary forms a local 2-design on $s$ qubits, we have
\begin{align}
    \Var[\partial_\nu \ell_i(\thv)] = r_{n,s} D_{HS} \Big(\rho_i^{(h)}, \frac{\emph{\id}}{2^s}\Big) \,,
\end{align}
where $r_{n,s} \in \Omega(1/\poly (n))$, and where we denote $D_{HS} (A,B) = \Tr[(A-B)^2]$  as the Hilbert-Schmidt distance. Here, $\frac{\emph{\id}}{2^s}$ is the the maximally mixed state on $s$ qubits and  $\rho_i^{(h)}=\Tr_{\overline h}[\rho_i]$ is the reduced input state on the $s$ qubits acted upon by the $h$-th unitary.
\end{proposition}

First, let us remark that Proposition~\ref{prop:local-cost} shows that a standard VQA that uses the cost in Eq.~\eqref{eq:cost-local} (tensor product ansatz and local measurement) does not exhibit a BP. This is due to the fact that when the input state to the VQA is $\ket{0}^{\otimes n}$, then $D_{HS} (\dya{0}^{\otimes s},\frac{{\id}}{2^s}) =1-1/2^s$, and hence $\Var[\partial_\nu \ell_i(\thv)]\in\Omega(1/\poly (n))$ (assuming $s$ does not scale with $n$). 

However, for a QML setting, Proposition \ref{prop:local-cost}, combined with Theorem \ref{thm:var-generic-loss}, implies that the QML  model is susceptible to dataset-induced BPs  even with simple QNNs and local measurements.  Specifically, if the reduced input state is exponentially close to the maximally mixed state then one has a BP. Examples of such datasets are Haar-random quantum states~\cite{brandao2010hastings} or classical data encoded with a scrambling unitary~\cite{holmes2021barren}, as the reduced states can be shown to concentrate around the maximally mixed state through Levy's lemma~\cite{ledoux2001concentration}. 

We note that this phenomenon is similar to entanglement-induced BPs~\cite{marrero2020entanglement,sharma2020trainability}. However, in a typical entanglement-induced BP setting, it is the QNN that generates high levels of  entanglement in its output states and thus leads to trainability issues. In contrast, here it is the quantum states obtained from the dataset that already contain large amounts of entanglement even before getting to the QNN.

It is important to make a distinction between the dataset-induced BPs for classical and quantum data. Specifically, special care must be taken when using classical datasets as here one can actually choose the embedding scheme, and this choice affects the amount of entanglement that the states $\rho_i$ will have. Such a choice is typically not present for quantum datasets. 

For classical datasets, let us make the important remark that (in a practical scenario) the embedding is not able to prevent a BP which would otherwise exist in the model.  As discussed in Section~\ref{sec:framework-QML}, the embedding process simply assigns a quantum state $\rho_i=\EC_{\vec{x}_i}^E(\dya{\vec{0}})$ to each classical data point $x_i$. Thus, previously established BP results that hold independently of the input state  such as global cost functions~\cite{cerezo2020cost},  deep circuits~\cite{mcclean2018barren}, expressibility~\cite{holmes2021connecting} and noise~\cite{wang2020noise} will hold regardless of the embedding strategy.

Our previous results further illuminate that the choice of embedding strategy is a crucial aspect of QML models with classical data, as it can affect that model's trainability. As argued in~\cite{havlivcek2019supervised}, a necessary condition for obtaining quantum advantage is that inner products of data-embedded states $\rho_i=\EC_{\vec{x}_i}^E(\dya{\vec{0}})$ should be classically hard to estimate. We note that of course this however is not sufficient to guarantee that the embedding is \emph{useful}. For instance, for a classification task it does not guarantee that states are embedded in distinguishable regions of the Hilbert space~\cite{lloyd2020quantum}. Thus, currently, one has the following criteria on which to design encoders for QML:
\begin{enumerate}
    \item Classically hard to simulate.
    \item Practical usefulness.
\end{enumerate}
From the results presented here,  a third criterion should be carefully considered when designing embedding schemes: \begin{enumerate}
\setcounter{enumi}{2}
    \item Not inducing trainability issues.
\end{enumerate}
This opens up a new research direction of trainability-aware quantum embedding schemes. 

Here we briefly note that while Proposition~\ref{prop:local-cost} was presented for a tensor product ansatz, a similar result can be obtained for more general QNNs such as the hardware efficient ansatz~\cite{cerezo2020cost} or the quantum convolutional neural network~\cite{pesah2020absence}. For these architectures it is found that $\Var[\partial_\nu \ell_i(\thv)]$ is upper bounded by a quantities such as $D_{HS}\Big(\rho_i^{(h)}, \frac{\emph{\id}}{2^s}\Big)$. However, while the form of the upper bound is more complex and cumbersome to report, the dataset-induced BP will arise for these architectures. 

\subsection{Empirical FI matrix in the presence of a barren plateau}\label{sec:fi-matrix}

Recently, the eigenspectrum of the empirical FI matrix was shown to be related to gradient magnitudes of the QNN loss function~\cite{abbas2020power}. Here we investigate this connection in more detail. Namely, we discuss how a BP in the QML model affects the empirical FI matrix and natural gradient-based optimizers (which employ the FI matrix).

In what follows, we show that under the conditions for which the negative log-likelihood loss gradients vanish exponentially, the matrix elements of the empirical FI matrix $\tilde{F}(\thv)$, as defined in Eq.~\eqref{eq:FI_empirical}, also vanish exponentially. This result is complementary to and extends the results in~\cite{abbas2020power}.  First, consider the following proposition.

\begin{proposition}\label{prop:FI}
Under the assumptions of Corollary \ref{corollary:cor1} for which the negative log-likelihood loss function has a BP according to Eq.~\eqref{eq:var-BP}, and assuming that the number of trainable  parameters in the QNN is in $ \OC(\poly(n))$, we have
\begin{equation}
    \mathbb{E}\big[\big|\tilde{F}_{\mu\nu}(\thv)\big|\big] \leq G(n) \,, \; \textrm{with}\;\, G(n) \in \OC(1/\alpha^n)\,,
\end{equation}
where $\tilde{F}_{\mu\nu}(\thv)$ are the matrix entries of $\tilde{F}(\thv)$ as defined in Eq.~\eqref{eq:FI_empirical}.
\end{proposition}

While Proposition~\ref{prop:FI}  shows that the expectation value of the FI matrix elements vanish exponentially, this result is not enough to guarantee that (with high probability) the matrix elements will concentrate around their expected value. Next we present a stronger concentration result which is valid for QNNs where the parameter shift rule holds~\cite{mitarai2018quantum,schuld2019evaluating}.

\begin{corollary}\label{corollary:cor2}
Under the assumptions of Corollary \ref{corollary:cor1} for which the negative log-likelihood loss function has a BP according to Eq.~\eqref{eq:var-BP}, and assuming that the QNN structure allows for the use of the parameter shift rule, we have
\begin{equation}
    \emph{Pr}\left(\left|\tilde{F}_{\mu\nu} - \mathbb{E}[\tilde{F}_{\mu\nu}] \right| \geq c\right) \leq \frac{Q(n)}{c^2}\,,
\end{equation}
where $\mathbb{E}[\tilde{F}_{\mu\nu}]\leq H(n)$ and $Q(n),H(n)\in \OC(1/\alpha^n)$.
\end{corollary}

We note that these two results imply that when the linear expectation values exhibit a BP, then one requires an exponential number of shots to estimate the entries of the  FI matrix, and concomitantly its eigenvalues. This also implies that optimization methods that use the FI, such as natural gradient methods, cannot navigate the flat landscape of the BP (without spending exponential resources).

\section{Numerical Results}
In this section, we present numerical results studying the trainability of QNNs in supervised binary classification tasks with classical data. Specifically, we analyze the effect of the dataset, the embedding, and the locality of the measurement operator on the trainability of QML models. In what follows, we first describe the dataset, embedding, QNN, measurement operators and loss functions used in our numerics.

We consider two different types of datasets, one composed of structured data and the other one of unstructured random data. This allows us to further analyze how the data structure can affect the trainability of the QML model. First, the structured dataset is formed from handwritten digits from the MNIST dataset~\cite{lecun1998mnist}. Here, greyscale images of "0" and "1" digits are converted to length-$n$ real-valued vectors $\vec{x}$ using a principal component analysis method~\cite{jolliffe2005principal} (see also Appendix \ref{sec:appdx-numerics} for a detailed description). Then, for the unstructured dataset we  randomly generate vectors $\vec{x}$ of length $n$ by uniformly sampling each of their components from $[-\pi,\pi]$. In addition,  each random data point is randomly assigned a label. 

\begin{figure}[t]
	\includegraphics[width= 1 \columnwidth]{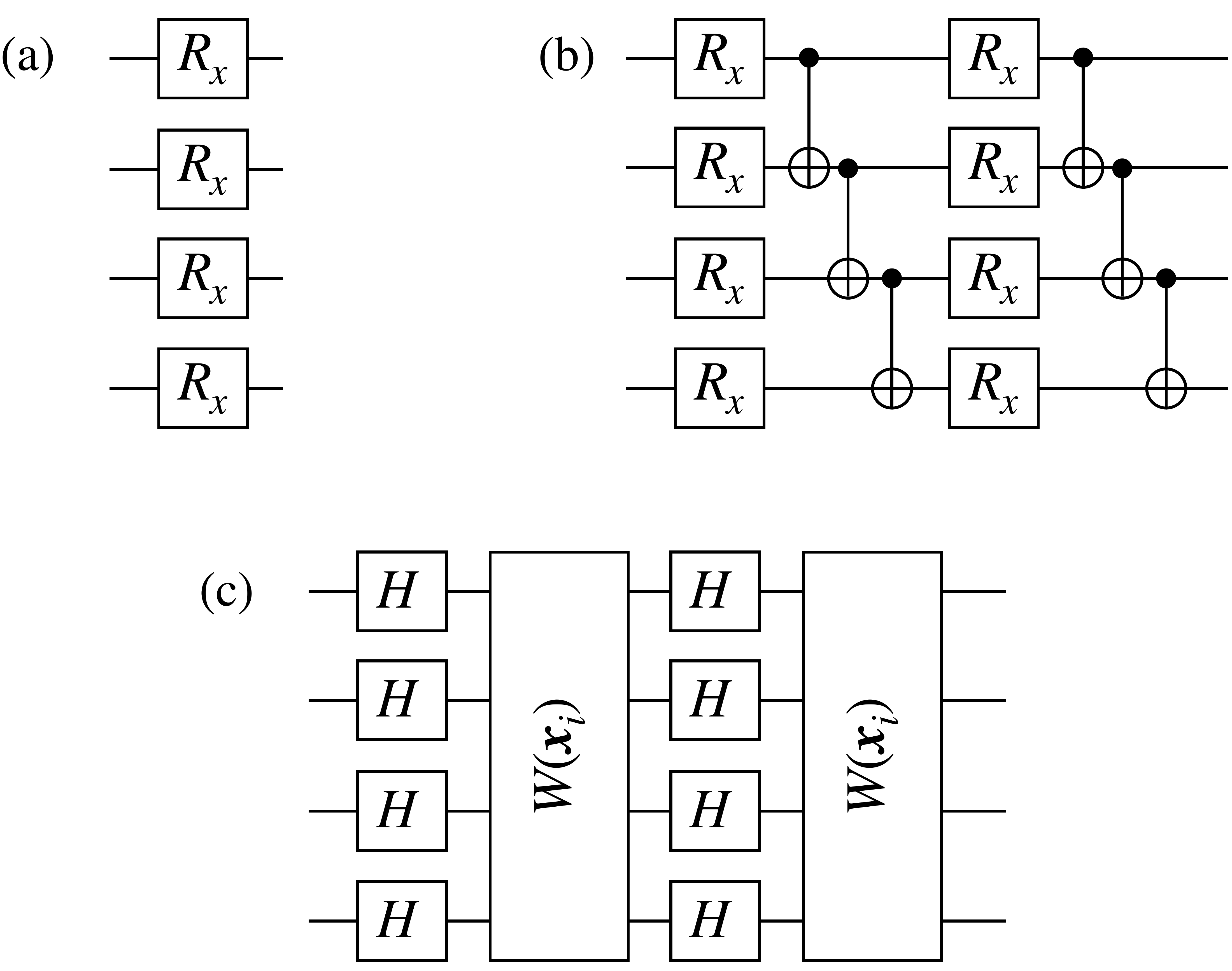}
	\caption{\textbf{Circuits for embedding unitaries used in our numerics.} (a) Tensor Product Embedding (TPE), composed of single qubit rotations around the $x$-axis. Here, the encoded state $\rho_i$ is obtained by applying a rotation  $R_x$ on the $j$-th qubit whose angle is the  $j$-th component of the vector $\vec{x}_i$. (b) Hardware Efficient Embedding (HEE). A  layer is composed of single qubit rotation around the $x$-axis whose rotation angles are assigned in the same way as in the TPE. After each layer of rotations, one applies entangling gates acting on adjacent pairs of qubits. (c) Classically Hard Embedding (CHE). Each unitary $W(\vec{x}_i)$ is composed of  single- and two-qubit gates that are diagonal in the computational basis. } \label{fig:embedding}
\end{figure}

In all numerical settings that we study, the embedding $\EC_{\vec{x}_i}^E$ is a unitary acting on $n$ qubits, which we denote as $V(\vec{x}_i)$. Thus, the output state of the embedding is a pure state of the form $\ket{\psi(\vec{x}_i)}=V(\vec{x}_i)\ket{\vec{0}}$. As shown in Fig.~\ref{fig:embedding}, we use three different embedding schemes for $V(\vec{x}_i)$. The first is the Tensor Product Embedding (TPE). The TPE is composed of single-qubit rotations around the $x$-axis so that $\ket{\psi(\vec{x}_i)}$ is obtained by applying a rotation  on the $j$-th qubit whose angle is the  $j$-th component of the vector $\vec{x}_i$. The second embedding scheme is presented in Fig.~\ref{fig:embedding}(b), and is called the Hardware Efficient Embedding (HEE). In a single layer of the HEE, one applies rotations around the $x$-axis followed by entangling gates acting on adjacent pairs of qubits. Finally, we refer to the third scheme as the Classically Hard Embedding (CHE). The CHE was proposed in~\cite{havlivcek2019supervised}, and is based on the fact that the inner products between output states of $V(\vec{x}_i)$ are believed to be hard to classically simulate as the depth and width of the embedding circuit increases. The unitaries  $W(\vec{x}_i)$ in each layer of the CHE are composed of  single- and two-qubit gates that are diagonal in the computational basis. We refer the reader to Appendix~\ref{appx:che-architecture} for a description of the unitaries $W(\vec{x}_i)$.

For the QNN in the model we consider two different ansatzes. The first is composed of a single layer of parametrized single-qubit rotations $R_y$ about the $y$-axis. Here, the output of the QNN is an $n$-qubit state. The second QNN we consider is the Quantum Convolutional Neural Network (QCNN) introduced in~\cite{cong2019quantum}. The QCNN is composed of a series of convolutional and pooling layers that reduce the dimension of the input state while preserving the relevant information. In this case, the output of the QNN is a $2$-qubit state. We refer the reader to  Appendix~\ref{appx:qcnn-architecture} for a more detailed description of the QCNN we used.

When using the QNN composed of $R_y$ rotations, we apply a global measurement on all qubits to compute the expectation value of the global operator $Z^{\otimes n}$. Thus, the predicted label and label probabilities are given by Eqs.~\eqref{eq:label-pred}--\eqref{eq:prob-parity}. Since global measurements are expected to lead to trainability issues, we also propose a local measurement  where one computes the expectation value of $Z^{\otimes2}$ over the reduced state of the middle two qubits. Note that this is equivalent to computing the parity of the length-$2$ output bitstrings. On the other hand, when using the QCNN, one naturally has to perform a local measurement as here the output is a $2$-qubit state. Thus, here we also compute the expectation value of $Z^{\otimes2}$.

Finally, we note that in our numerics we study the scaling of the gradient of both the mean-squared error and log-likelihood loss functions of Eqs.~\eqref{eq:loss-mse}  and Eq.~\eqref{eq:loss-log}, respectively. Moreover, we also consider the scaling of the gradients of the linear expectation value in Eq.~\eqref{eq:measurement}.

\subsection{Global measurement}

\begin{figure}[t!]
	\includegraphics[width= 1 \columnwidth]{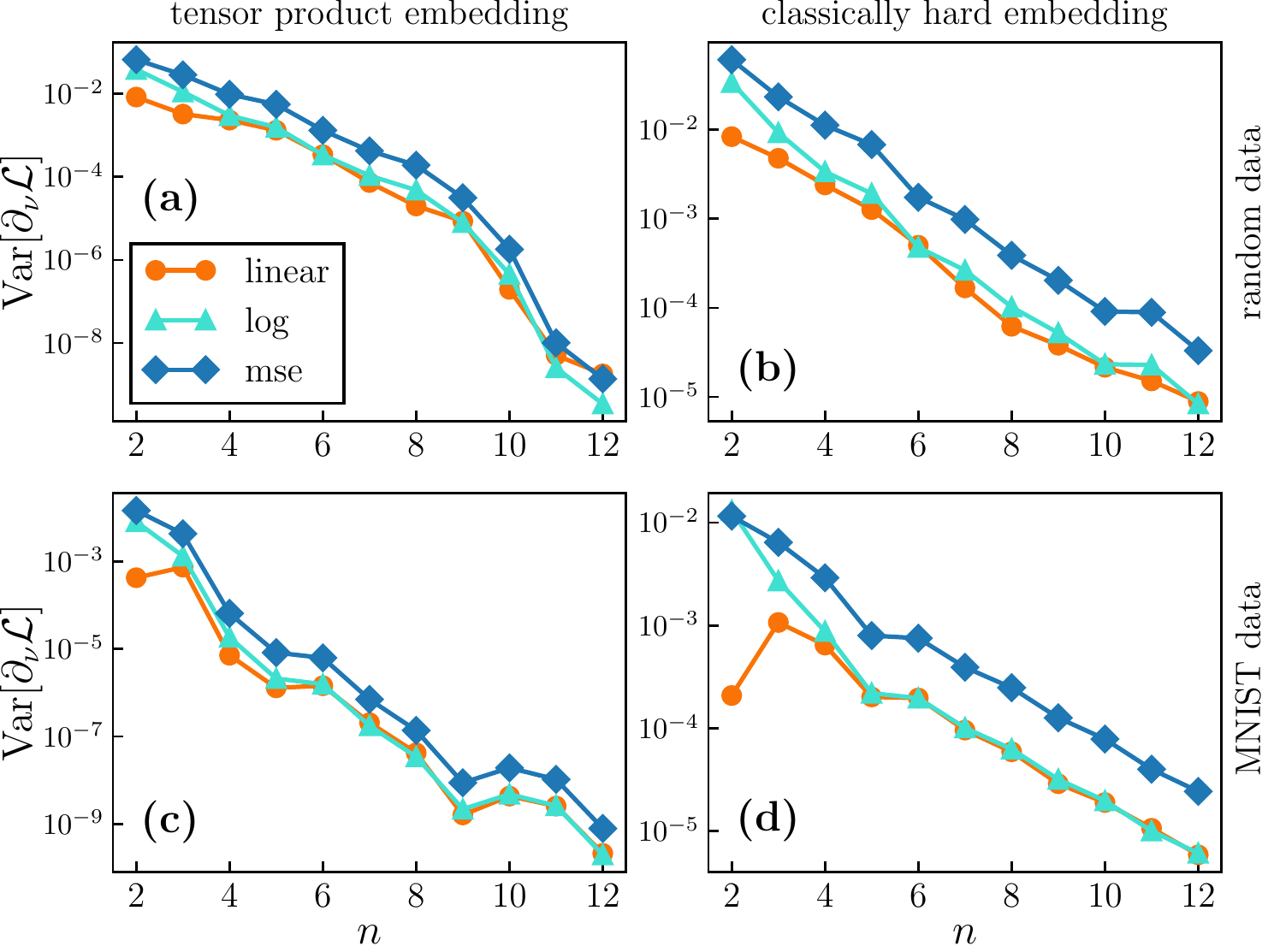}
	\caption{\textbf{Variance of the  partial derivative versus number of qubits.} Here we consider linear expectation, log-likelihood and mean squared error loss functions with global measurement. In (a)-(b) we consider an unstructured (random)  datasets, while in (c)-(d) a structured (MNIST) dataset. The classical data is encoded via the TPE (a)-(c) and the CHE (b)-(d) schemes. We plot the variance of the partial derivative versus number of qubits for all loss functions. The partial derivative is taken over the first parameter of the tensor product QNN. } \label{fig:lin-log-quad-MNIST}
\end{figure}

Here we first study the effect of performing global measurements on the output states of the QNN. As shown in~\cite{cerezo2020cost}, we expect that linear functions of global measurements will have BPs. For this purpose we consider both structured (MNIST) and unstructured (random) datasets encoded through the TPE and CHE schemes (see Fig.~\ref{fig:embedding}). The QNN is taken to be the tensor product of single qubit rotations and we measure the expectation value of $Z^{\otimes n}$.

Figure~\ref{fig:lin-log-quad-MNIST} presents results where we numerically compute the variance of partial derivatives of the linear expectation values in~\eqref{eq:measurement}, the mean squared error loss function in Eq.~\eqref{eq:loss-mse}, and the negative log-likelihood loss function in Eq.~\eqref{eq:loss-log} versus the number of qubits $n$. The variance is evaluated over 200 random sets of QNN parameters and the dataset is composed of $N=10n$ data points. 

Figure~\ref{fig:lin-log-quad-MNIST} shows that, as expected from~\cite{cerezo2020cost}, the variance of the partial derivative of the linear loss function vanishes exponentially with the system size, indicating the presence of a BP according to Eq.~\eqref{eq:var-BP}. Moreover, the plot shows that for all dataset and embedding schemes considered, the variance of the partial derivatives of the log-likelihood and mean squared error loss functions also vanish exponentially with the system size. This scaling is in accordance with the results of Corollary~\ref{corollary:cor1}.

\begin{figure}[t]
\includegraphics[width=.99\columnwidth]{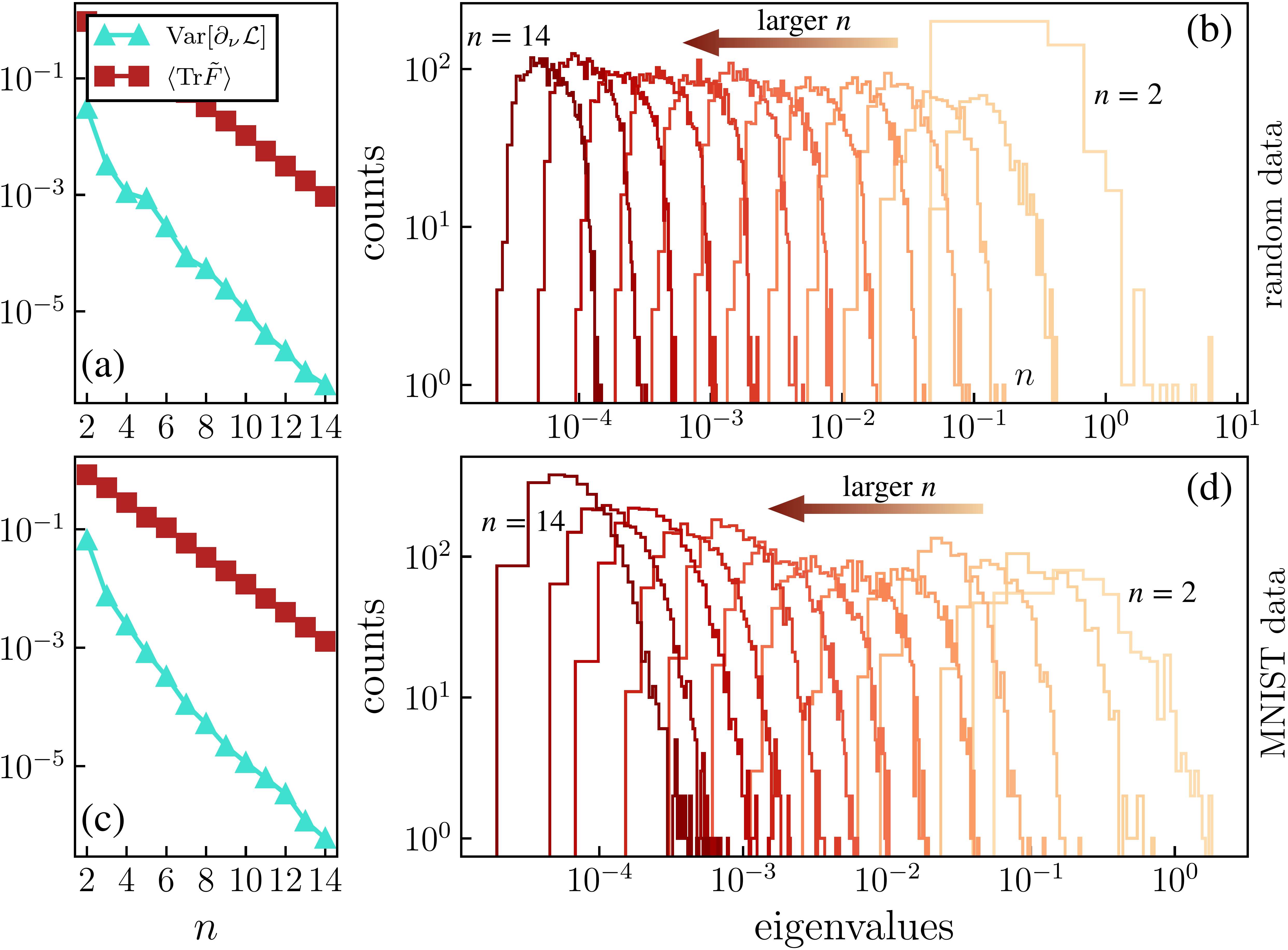}	
	\caption{\textbf{Trace and spectrum of the Fisher Information matrix in a BP.} The top panels correspond to an unstructured (random) dataset, with the bottom to a  structured (MNIST) dataset. In both cases the data is encoded using the  CHE scheme and then sent through  the tensor product QNN with a global measurement. In (a) and (c), we plot the variance of the partial derivatives of the log-likelihood loss function, and the expectation value of the trace of the empirical FI matrix versus the number of qubits $n$. Here, the expectation values are taken over $200$ different sets of QNN parameters.   In (b) and (d), we show the eigenvalues distribution of the empirical FI matrix for increasing numbers of qubits. }
	\label{fig:global_var_trace_fi_IQP}
\end{figure}

Here we note that the presence of BPs can be further characterized through the spectrum of the empirical FI matrix~\cite{abbas2020power}. As mentioned in Section~\ref{sec:fi-matrix}, in a BP the magnitudes of the eigenvalues of the empirical FI matrix  will decrease exponentially as the number of qubits increases.  In Fig.~\ref{fig:global_var_trace_fi_IQP}(a) and (c) we plot the trace of the empirical FI matrix versus the number of qubits for the same structured and unstructured datasets of Fig.~\ref{fig:lin-log-quad-MNIST}. As expected, the trace decreases exponentially with the problem size when using a global measurement due to the loss function exhibiting a BP. While the trace of the empirical FI provides a coarse-grained study of the eigenvalues, we also show in   Fig.~\ref{fig:global_var_trace_fi_IQP}(b) and (d) representative results for the eigenvalue spectrum distributions of the  empirical FI matrix. One can see here that all eigenvalues become exponentially vanishing with increasing system size.

Our results here show that even for a trivial QNN, and independently of the dataset and embedding scheme, global measurements in the loss function lead to exponentially small gradients, and thus to BPs in QML models. Moreover, we have also verified that the eigenvalues of the empirical FI matrix are, as expected from Corollary~\ref{corollary:cor2}, exponentially small in a BP, showing that an exponential number of shots are needed to accurately estimate  the matrix elements, eigenvalues, and trace of  the empirical FI matrix. This precludes the possibility of efficiently estimating quantities such as the normalized empirical FI matrix $\tilde{F}(\thv)/\Tr[\tilde{F}(\thv)]$~\cite{abbas2020power}.

\begin{figure}[t]
	\includegraphics[width= 1 \columnwidth]{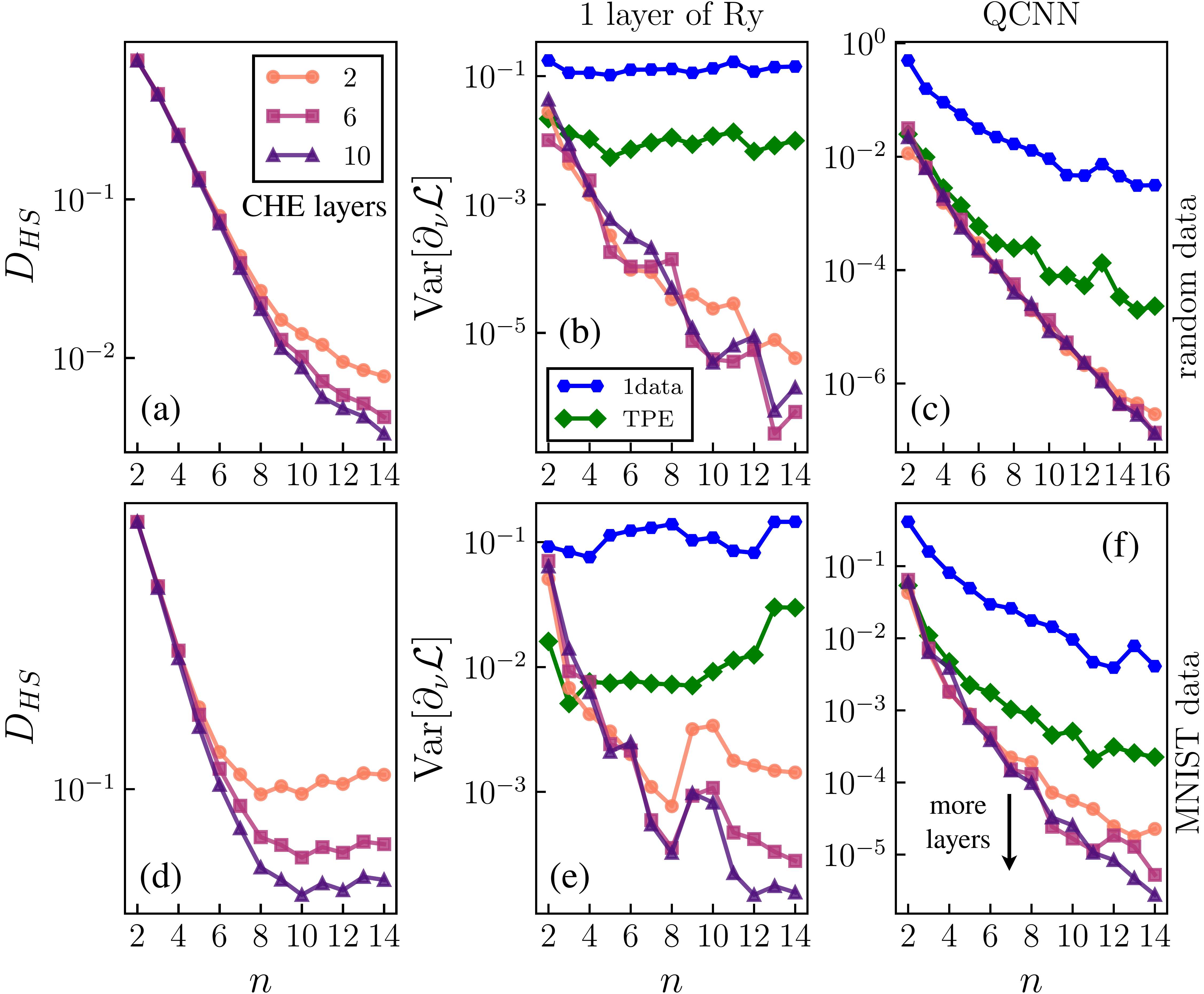}
	\caption{\textbf{Effect of CHE scheme with local measurement on trainability.} The top panels correspond to an unstructured (random) dataset, with the bottom to a  structured (MNIST) dataset. In all cases we used the CHE scheme with different number of layers. Panels  (a) and (d), show the Hilbert-Schmidt distance $D_{HS}(\rho^{(2)}_{i},\id/4)$ versus the number of qubits. Here,  $\rho^{(2)}_{i}$ is the reduced state of the central two qubits. Panels (b) and (e) show  the variances of the partial derivative of the log-likelihood loss function versus the number of qubits $n$ for the tensor product QNN, with panels (c) and (f)  for a QCNN. We also plot as reference the variances using the non-entangling TPE scheme when the loss is evaluated over the dataset and over a single data point. }\label{fig:local_var_trace_IQP}
\end{figure}

\subsection{Dataset and embedding-induced barren plateaus}\label{sec:numerics-embedding}

Here we numerically study how the embedding scheme and the dataset can potentially lead to trainability issues. Specifically, we recall from Section \ref{sec:embedding-induced-BP} that highly-entangling embedding schemes can lead to reduced sates being concentrated around the maximally mixed state, and thus be harder to train local gates on. To check how close reduced states at the output of the CHE and HEE schemes are, we average the Hilbert-Schmidt distance $D_{HS}(\rho_i^{(2)},\id/4)$ between the maximally mixed state and the reduced state $\rho_i^{(2)}$ of the central two qubit. In addition, we further average over 2000 data points from a structured (MNIST) and unstructured (random) dataset. 

Results are shown in Fig.~\ref{fig:local_var_trace_IQP}(a) and (d), where we plot  $D_{HS}(\rho_i^{(2)},\id/4)$ versus the number of qubits for the CHE scheme with different number of layers. As expected, here we see that increasing the number of layers in the embedding leads to higher entanglement in the encoded states $\rho_i$, and thus to reduced states being closer to the maximally mixed state. Moreover, here we note that the structure of the dataset also plays a role in the mixedness of the reduced state, as the Hilbert-Schmidt distances $D_{HS}(\rho_i^{(2)},\id/4)$ for the unstructured random dataset can be up to one order of magnitude smaller than those for structured dataset.

\begin{figure}
	\includegraphics[width= 1 \columnwidth]{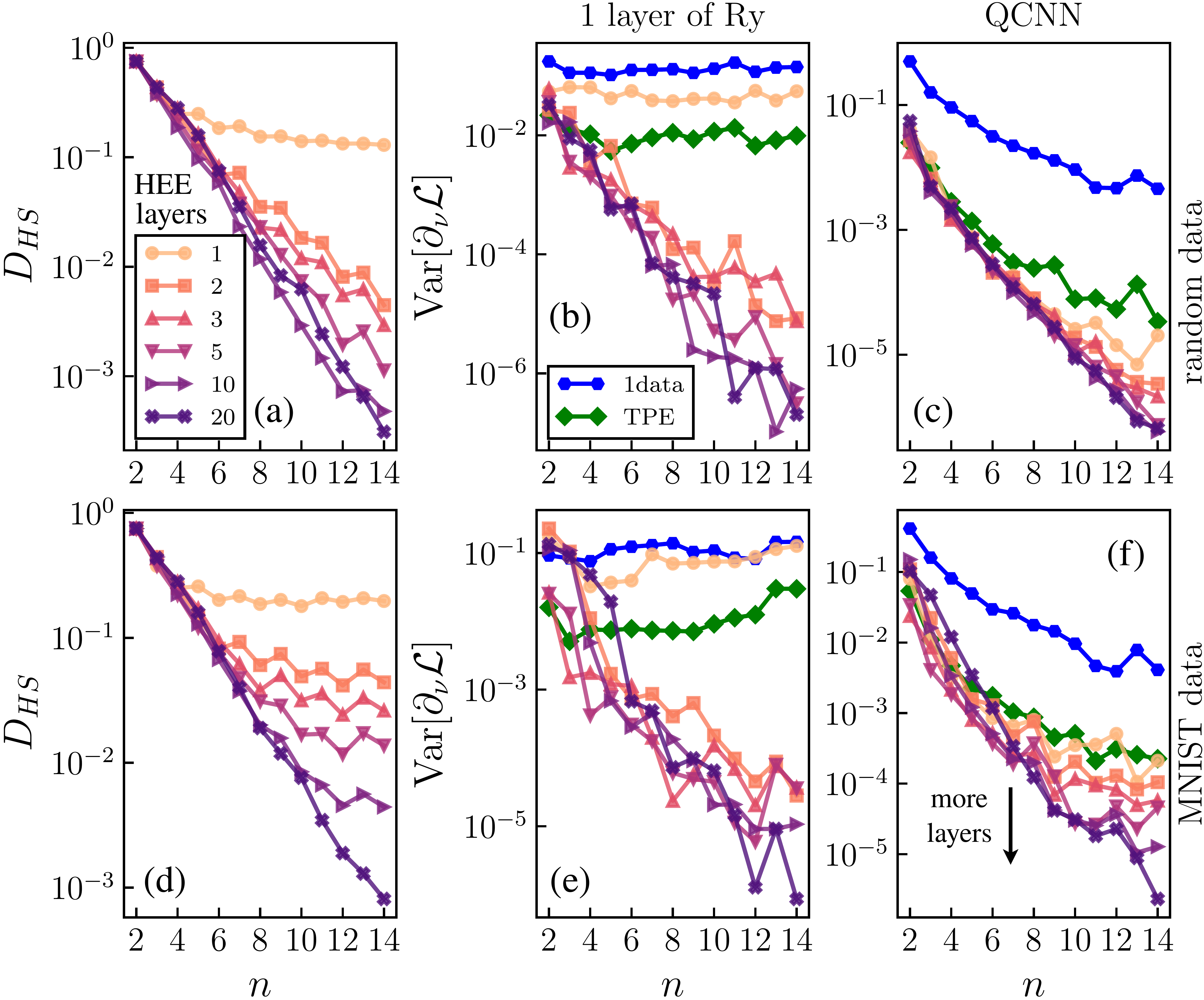}
	\caption{\textbf{Effect of HEE scheme with local measurement on trainability.}  The top panels correspond to an unstructured (random) dataset, with the bottom to a  structured (MNIST) dataset. In all cases we use the HEE scheme with increasing number of layers. Panels  (a) and (d), show the Hilbert-Schmidt distance $D_{HS}(\rho_i^{(2)},\id/4)$ versus the number of qubits. Here,  $\rho^{(2)}_{i}$ is the reduced state of the central two qubits. Panels (b) and (e) show  the variances of the partial derivative of the log-likelihood loss function versus the number of qubits $n$ for the tensor product QNN, with panels (c) and (f)  for a QCNN. We also plot as reference the variances using the non-entangling TPE scheme when the loss is evaluated over the dataset and over a single data point.}
	\label{fig:local_var_trace_HEA}
\end{figure}

In Fig.~\ref{fig:local_var_trace_IQP} we also show the variance of the log-likelihood loss function partial derivative as a function of the number of qubits and the number of CHE layers. Here we use both the tensor product QNN (panels (b) and (e)) and the QCNN (panels (c) and (f)), and we compute local expectation values of $Z^{\otimes 2}$. Moreover, here the dataset is composed of $N=10n$ points, and the variance is taken by averaging over $200$ sets of random QNN parameters. Since both the tensor product QNN with local cost and the QCNN  are not expected to exhibit BPs with no training data and separable input states (see~\cite{cerezo2020cost} and~\cite{pesah2020absence}, respectively), any unfavorable scaling arising here will be due to the structure of the data or the embedding scheme. To ensure this is the case, we plot two additional quantities as references. The first is obtained for the case when the embedding scheme is simply replaced with the TPE, representing the scenario of a non-entangling encoder. In the second, we also use the TPE encoder but rather than computing the loss function over the whole dataset, we only compute it over a single data point, i.e.  $\LC_{log}(\thv)=-\log p_i$. Then, for this single-data point loss function, we study the scaling of the partial derivative variance and we finally  average over the dataset, i.e. $\sum_i \Var[\partial_\nu \log p_i]/N$. 
This allows us to characterize the effect of the size and the randomness associated with the dataset.

For the unstructured random dataset, we can see that $\Var[\partial_\nu \LC(\thv)]$ appears to vanish exponentially for both QNNs and for all considered number of layers in the CHE. This shows that the randomness in the dataset ultimately translates into randomness in the loss function and in the presence of a BP. For the structured dataset, we can see that $\Var[\partial_\nu \LC(\thv)]$ does not exhibit an exponentially vanishing behaviour for small number of CHE layers. However, as the depth of the embedding increases, the variances become smaller. In particular, when using a QCNN, increasing the number of CHE layers appears to change the behaviour of $\Var[\partial_\nu \LC(\thv)]$ towards a more exponentially vanishing scaling. Finally, we observe that the variance of the loss function constructed from a single data point is always larger than the loss constructed from $N$ data points. This indicates that the larger the dataset, the smaller the variance. 

To further study this phenomenon, in Fig.~\ref{fig:local_var_trace_HEA} we repeat the calculations of Fig.~\ref{fig:local_var_trace_IQP} but using the  HEE scheme instead. That is, we show the scaling of $D_{HS}(\rho_{k}(\vec{x}_i),\id/4)$ and $\Var[\partial_\nu \LC(\thv)]$ versus the number of qubits for the HEE scheme with different number of layers and for structured (MNIST) and unstructured (random) datasets. 

Here, the effect of the entangling power of the embedding on the Hilbert-Schmidt distance $D_{HS}(\rho_{k}(\vec{x}_i),\id/4)$   can be seen in panels (a) and (d) of Fig.~\ref{fig:local_var_trace_HEA}. Therein one can see that as the number of layers of the HEE increases (and thus also entangling power~\cite{holmes2021connecting}) the distance to the maximally mixed state vanishes exponentially with the system size. One can see here that, independently of the structure of the dataset, the large entangling power of the embedding scheme leads to states that are essentially maximally mixed on any reduced pair of qubits. As seen in Fig.~\ref{fig:local_var_trace_HEA}(b), (c), (e) and (f), the latter then translates into an exponentially vanishing $\Var[\partial_\nu \LC(\thv)]$ and thus a BP.

These results indicate that the choice of dataset and embedding method can have a significant effect on the trainability of the QML model. Specifically, QNNs that have no BPs when trained on trivial input states can have exponentially vanishing gradients arising from either the structure of the dataset, or the large entangling power of the embedding scheme. Moreover, these results show that  the  Hilbert-Schmidt distance can be used as an indicator  of how much the embedding can potentially worsen the trainability of the model.

\subsection{Practical usefulness of the embedding scheme and local measurements}

As discussed in Section~\ref{sec:embedding-induced-BP}, a good  embedding  satisfies (at least) the following three criteria: classically hard to simulate, practical usefulness, not inducing trainability issues. In the previous sections, we  have studied how the trainability of the model can be affected  by the embedding choice and dataset. Here we point out another subtlety, namely, that ``classically-hard-to-simulate'' and ``practical usefulness'', do not always coincide. 
Particularly, we here show that the CHE scheme can lead to poor performance for some standard benchmarking test. 

For this purpose we choose the QCNN architecture, which uses a local measurement and is known to not exhibit a BP~\cite{pesah2020absence}, to solve the task of classifying handwritten digits `0' and `1'  from the MNIST dataset (see also~\cite{hur2021quantum} for a similar study using QCNNs for MNIST classification). Then,  we compare two choices for the embedding scheme. The first is a classically simulable scheme given by a single layer of the  HEE, whilst the other is the (conjectured) classically hard to simulate two-layered CHE.

To make this comparison fair, both encoders are subjected to an identical setting: QCNN implemented on $n=8$ qubits, compute the expectation values of $Z^{\otimes 2}$, use the log-likelihood loss function, and have a  training and testing datasets of respective size 400 and 40. In all cases the classical optimizer used is ADAM~\cite{kingma2015adam} with a $0.02$ learning rate. For each training iteration, the expectation value measured from the QML model is fed into an additional classical node with a $tanh$ activation function. 

In Fig.~\ref{fig:CHE-HEE-qcnn-train}, we show the training loss functions and test accuracy versus the number of iterations for 10 different optimization runs using each encoders. We observe that, despite being classically simulable, the model with HEE has a significantly better performance (above $90\%$ test accuracy)  than the model with CHE (around $65\%$ accuracy) on both training and testing  for this specific task.  Hence, we we can see a particular embedding example where hard-to-simulate  does not translate into practical usefulness. 

We emphasize that this result should not be interpreted as the CHE scheme being generally unfavorable for practical purposes. Rather, that additional care should be taken when choosing encoders to suit specific tasks, and to highlight the challenge of designing encoders to satisfy all necessary criteria for achieving a quantum advantage.

\begin{figure}
 	\includegraphics[width= 1 \columnwidth]{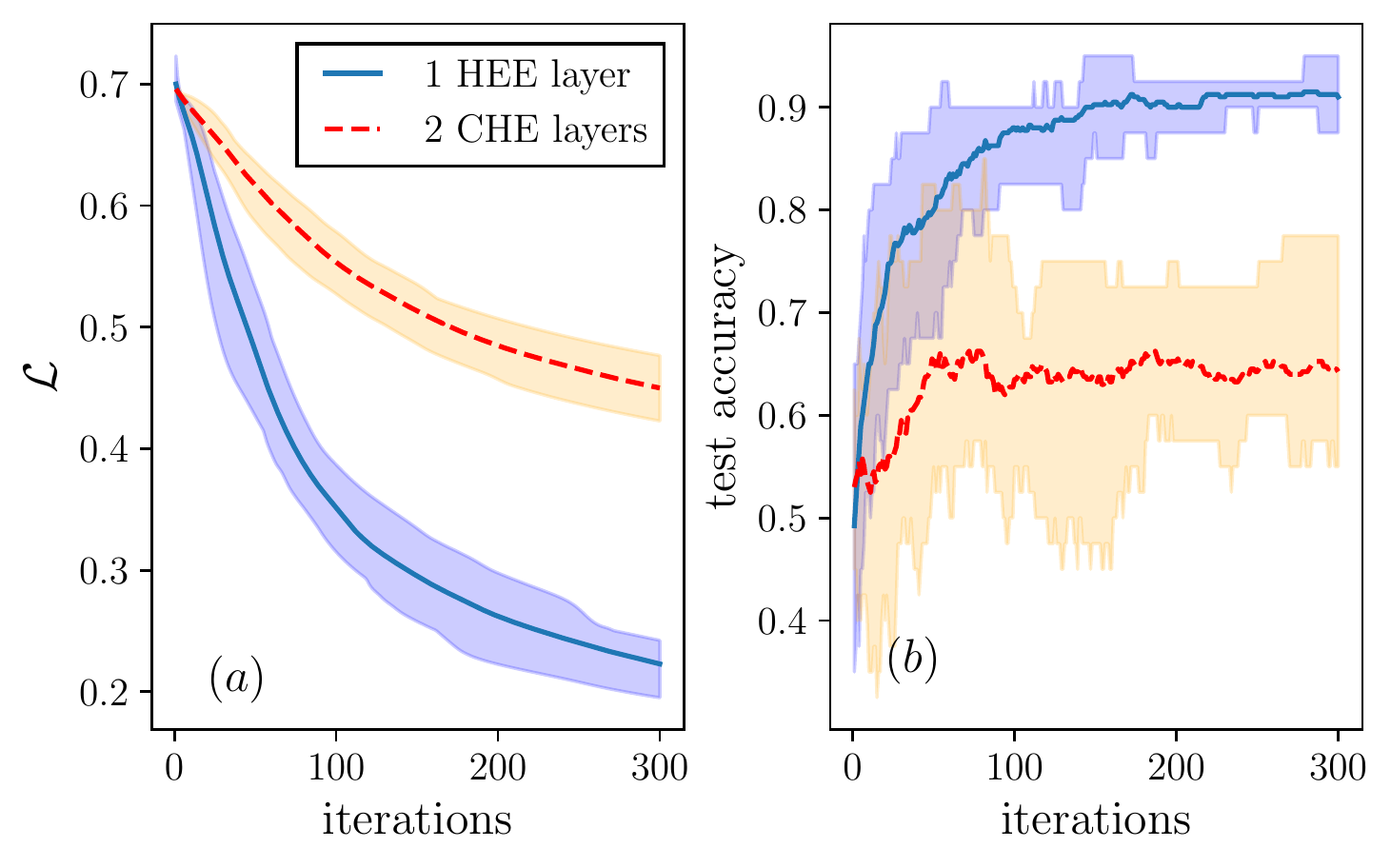}
 	\caption{\textbf{Loss function (a) and test accuracy (b) versus number of iterations for MNIST classification using a QCNN.}  We train the QML model with $n=8$ qubits using two different embedding schemes (1 HEE layer and 2 CHE layers) for the task of binary classification between digits `0' and `1' from the MNIST dataset. Solid lines represents the average over 10 instances of different initial parameters, and shaded areas represents the range over these instances.}\label{fig:CHE-HEE-qcnn-train}
 \end{figure}

\section{Implications for the literature}

Here we briefly summarize the implications of our results for the QML literature, specifically the literature on training quantum neural networks. 

First, we have shown that features deemed detrimental to training linear cost functions in the VQA framework will also lead to trainability issues in QML models. This is particularly relevant to the use of global observables such as measuring the parity of the output bitstrings on all qubits, which have been employed in the QML literature. We remark that there is no \textit{a priori} reason to consider the global parity. One could instead just measure a subset of qubits and assign labels via local parities, or even average the local parities across subsets of qubits. As shown in our numerics section, local parity measurements are practically useful and one can use them to optimize the model and achieve small training and generalization errors.

Second, our results indicate that QML models can exhibit scaling issues due to the dataset. Specifically, when the input states to the QNN have large amounts of entanglement, then the QNN's parameters can become harder to train as local states will be concentrated around the maximally mixed state. This is particularly relevant when dealing with classical data, as here one has the freedom to choose the embedding. This points to the fact that the choice of embedding needs to be carefully considered and that trainability-aware encoders should be prioritized and developed. 

Unfortunately for the field of QML, data embeddings cannot solve BPs by themselves. In other words, as proven here, the choice of embedding cannot in practice mitigate the effect of BPs or prevent a BP that would otherwise exist for the QNN. For instance, the embedding cannot prevent a BP arising from the use of a global measurement, or from the use of a QNN that forms a 2-design. Hence, while embeddings can lead to a novel source of BPs, they cannot cure a BP that a particular QNN suffers from.

Finally, we show that optimizers relying on the FI matrix, such as natural gradient descent, require an exponential number of measurement shots to be useful in a BP. This is due to the fact that the matrix elements of the empirical FI matrix are exponentially small in a BP. Hence, quantities such as the normalized empirical FI matrix, which has been employed in the literature, are also inaccessible without incurring in an exponential cost.

\section{Discussion}

Quantum Machine Learning (QML) has received significant attention due to its potential for accelerating data analysis using quantum computers. A near-term approach to QML is to train the parameters of a Quantum Neural Network (QNN), which consists of a parametrized quantum circuit, in order to minimize a loss function associated with some data set. The data can either be quantum or classical (as shown in Fig.~\ref{fig:1}), with classical data requiring a quantum embedding map. While this novel, general paradigm for data analysis is exciting, there are still very few theoretical results studying the scalability of QML. The lack of theoretical results motivates our work, in which we focused on the trainability and gradient scaling of QNNs.

In the context of trainability, most of the previous results have been derived for the field of Variational Quantum Algorithms (VQAs). While VQAs and QML models share some similarities in that both train parametrized quantum circuits, there are some key differences that make it difficult to directly apply VQA trainability results to the QML setting. In this work, we bridged the gap between the VQA and QML frameworks by rigorously showing that gradient scaling results from VQAs will hold in a QML setting. This involved connecting the gradients of linear cost functions to those of mean squared error and log-likelihood cost functions.

In light of our results, many QML proposals in the literature would need to be revised, if they aim to be scalable. For instance, we rigorously proved that  features deemed detrimental for VQAs, such as global measurements or deep unstructured (and thus highly entangling) ansatzes, should also be avoided in QML settings. These results hold regardless of the data embedding, and hence one cannot expect  the data embedding to solve a barren plateau issue associated with a QNN.

Moreover, due to the use of datasets, we discovered a novel source for barren plateaus in QML loss functions. We refer to this as a Dataset-Induced Barren Plateau (DIBP). DIBPs are particularly relevant  when dealing with classical data, as here additional care must be taken when choosing an embedding scheme. A poor embedding choice could lead to a DIBP. Until now, a ``good'' embedding  was one that is classically hard to simulate and practically useful. However, our results show that a third criterion must be added for the encoder: not inducing gradient scaling issues. This paves the way towards the development of trainability-aware embedding schemes. 

Our numerical simulations verify the DIBP phenomenon, as therein we show how the gradient scaling can be greatly affected both by the structure of the dataset, as well as by the choice of the embedding scheme. Furthermore, our results illustrate another subtlety that arises when using classical data, as the classically-hard-to-simulate embedding of~\cite{havlivcek2019supervised} leads to large generalization error on a standard MNIST classification task. Thus, ``classically-hard-to-simulate'' and ``practical usefulness'' of an encoder do not always coincide.

Taken together, our results illuminate some subtleties in training QNNs in QML models, and show that more work needs to be done to guarantee that QML schemes will be trainable, and thus useful for practical applications.

\section{Acknowledgements}


We thank Michael Grosskopf, Julia Nakhleh, Amira Abbas, Christa Zoufal, and David Sutter for helpful discussions. ST and NAG were supported by the U.S. DOE through a quantum computing program sponsored by the Los Alamos National Laboratory (LANL) Information Science \& Technology Institute. ST was also supported by the National Research Foundation, Prime Minister's Office, Singapore and the Ministry of Education, Singapore under the Research Centres of Excellence programme.  SW was supported by the UKRI EPSRC grant no.~EP/T001062/1 and the Samsung GRP grant. PJC acknowledges initial support from the Laboratory Directed Research and Development (LDRD) program of LANL under project number 20190065DR, as well as support by the U.S. DOE, Office of Science, Office of Advanced Scientific Computing Research, under the Accelerated Research in Quantum Computing (ARQC) program. MC acknowledges initial support from the Center for Nonlinear Studies at LANL, as well as support from the LDRD program of LANL under project number 20210116DR. Research presented in this article was supported by the NNSA’s Advanced Simulation and Computing Beyond Moore’s Law Program at LANL.

\bibliography{quantum.bib}
\clearpage
\newpage
\onecolumngrid

\appendix
\vspace{0.5in}
\setcounter{theorem}{0}
\setcounter{proposition}{0}
\setcounter{corollary}{0}

\begin{center}
	{\Large \bf Appendix} 
\end{center}

In this Appendix we present further details for the  results of the main text. In Section \ref{sec:appdx-supplems} we present some supplemental lemmas. In Section \ref{sec:appdx-theorem} we present a proof for our main result Theorem \ref{thm:var-generic-loss}. In Section \ref{sec:appdx-cor} we present a proof of Corollary \ref{corollary:cor1} where we discuss the implications of our main result for the mean squared error and negative log-likelihood loss functions, and we discuss extensions to the corollary to other machine learning settings in Section \ref{sec:appdx-cor-extensions}. Finally, in Section \ref{sec:appdx-numerics} we provide additional details on our numerical simulations not presented in the main text.

\section{Supplemental lemmas}\label{sec:appdx-supplems}

In this section we present some supplemental lemmas that will be useful in deriving our results.

\begin{suplemma}[Variance of sum of random variables]\label{lem:var-sum}
Given a set of correlated random variables $\{X_i \}_i$, we have 
\begin{align}
    \Var\Big[\sum_i X_i\Big] \leq  \left( \sum_{i}\sqrt{\Var[X_i]} \right)^2\,.
\end{align}
\end{suplemma}
\begin{proof}
The variance of the sum of two correlated random variables is given by
\begin{align} \label{bound-var-sum}
    {\rm Var}[X_1 + X_2] & =  {\rm Var}[X_1] +  {\rm Var}[X_2] +  2{\rm Cov}[X_1,X_2]\,.
\end{align}
Using induction along with the fact that ${\rm Cov}(X_1+X_2,X_3) = {\rm Cov}(X_1,X_3) + {\rm Cov}(X_2,X_3) $, the variance of the full sum can be bounded as 
\begin{align}
    \Var\Big[\sum_i X_i\Big] &= \sum_i \Var[X_i] + \sum_{i\neq j} \Cov[X_i,X_j]\\
    &\leq \sum _i \Var[X_i] + \sum_{i\neq j}\sqrt{\Var[X_i]\Var[X_j]}\\
    &= \left( \sum_{i}\sqrt{\Var[X_i]} \right)^2\,,
\end{align}
where the inequality in the second line comes from the Cauchy-Schwarz inequality.
\end{proof}

\begin{suplemma}[Variance of product] \label{lem:var-product}
Given two correlated random variables $X$ and $Y$, we have
\begin{align}
    \Var[XY] \leq 2\Var[X] |Y^2|_{max} + 2(\mathbb{E}[X])^2\Var[Y]\,, \label{eq:}
\end{align}
where $|Y^2|_{max}$ is the maximum possible value of $Y^2$ i.e. $|Z| _{max} =  \max\{ |Z| : \rm{Pr}(Z) > 0 \}$. 
\end{suplemma}
\begin{proof}
We have
\begin{align}
    \Var[X + Y] &= \Var[X] + \Var[Y] + 2\Cov[X,Y] \\
    &\leq \Var[X] + \Var[Y] + 2\sqrt{\Var[X]\Var[Y]} \\
    &\leq \Var[X] + \Var[Y] + \sqrt{\Var[X]\Var[X]} + \sqrt{\Var[Y]\Var[Y]} \\
    &= 2\Var[X] + 2\Var[Y]\,, \label{eq:var-sum}
\end{align}
where in the first inequality we have used Cauchy-Schwarz, and the second inequality comes from the rearrangement inequality. Now consider 
\begin{align}
    \Var[XY] &= \Var\big[(X-\mathbb{E}[X])Y + \mathbb{E}[X]Y\big]\\
    &\leq 2\Var\big[(X-\mathbb{E}[X])Y\big] + 2\Var\big[\mathbb{E}[X]Y\big]\\
    &\leq 2\mathbb{E}\big[ (X-\mathbb{E}[X])^2Y^2 \big] + 2(\mathbb{E}[X])^2\Var[Y]\\
    &\leq 2\mathbb{E}\big[ (X-\mathbb{E}[X])^2 \big] |Y^2 |_{max} + 2(\mathbb{E}[X])^2\Var[Y] \\
    &= 2\Var[X] |Y^2|_{max} + 2(\mathbb{E}[X])^2\Var[Y]\,,
\end{align}
where in the first inequality we have used Eq.~\eqref{eq:var-sum}, in the second inequality we have used the definition of the variance, and in the third inequality we have simply taken the maximum value for $Y^2$. 
\end{proof}

\section{Proof of Theorem~\ref{thm:var-generic-loss}: Generic Loss Function}\label{sec:appdx-theorem}
In this section we provide the proof of Theorem~\ref{thm:var-generic-loss} which connects the variance of the partial derivative of general loss functions to the variance of the partial derivative of linear expectation values (the quantum model outputs). For convenience, we restate the theorem here.

\begin{theorem}[Variance of partial derivative of generic loss function]\label{thm:var-generic-loss-SM}
Consider the partial derivative of the loss function $\LC(\thv)$ in Eq.~\eqref{eq:generic-loss-function} taken with respect to variational parameter $\theta_\nu \in\thv$. We denote this quantity as $\partial_\nu \LC(\thv) \equiv \partial \LC(\thv)/\partial\theta_\nu$. The following inequality holds \small
\begin{align} \label{eq:supp-var-generic-loss-SM}
    \Var[\partial_\nu \LC(\thv)] \leq \left(\frac{1}{N} \sum_i {g_i} \sqrt{ \Var[\partial_\nu \ell_i(\thv)]+ (\mathbb{E}[\partial_\nu \ell_i(\thv)])^2} \right)^2,
\end{align}
\normalsize
where we used $\ell_i(\thv)$ as a shorthand notation for $\ell_i(\thv;y_i)$ in Eq.~\eqref{eq:measurement}, and where the expectation values are taken over the parameters $\thv$. Moreover, here we defined
\begin{align}
    g_i = \sqrt{2}\max_{\ell_i}\left|\frac{\partial f}{\partial \ell_i}\right|\,, \label{eq:g_i-SM}
\end{align}
where $\frac{\partial f}{\partial \ell_i}$ is the $i$-th entry of the Jacobian $J_{\vec{\ell}}$ with $\vec{\ell}=(\ell_1,\ldots,\ell_N)$, and where we denote as $\max_{\ell_i}\left|\frac{\partial f}{\partial \ell_i}\right|$ the maximum value of the partial derivative of $f(\ell_i,y_i)$.
\end{theorem}

\begin{proof}
For convenience we recall the generic loss function defined in~\eqref{eq:generic-loss-function} takes the form
\begin{align}
    \LC(\thv) = \frac{1}{N} \sum_i f(\ell_i(\thv;y_i),y_i)\,.
\end{align}
The variance of the partial derivative of this generic loss with respect to $\theta_\nu $ can be bounded as
\begin{align}
    \Var[\partial_\nu  \LC] = & \Var\left[ \frac{1}{N} \sum_i \partial_\nu f(\ell_i,y_i)\right] \\
    \leq & \left(\frac{1}{N}\sum_{i} \sqrt{\Var[\partial_\nu f(\ell_i,y_i)]}\right)^2\,,\label{eq:var-d-generic-cost} \\
    = & \left(\frac{1}{N}\sum_{i} \sqrt{\Var[\partial_{\ell_i} f(\ell_i,y_i) \partial_{\nu} \ell_i ]}\right)^2\,, \label{eq:proof-thm1-chain-rule}\\
    \leq & \left(\frac{1}{N}\sum_{i} \sqrt{2\Var[\partial_{\nu} \ell_i ]| \partial_{\ell_i}f(\ell_i,y_i)|^2_{max}+2(\mathbb{E}[\partial_\nu l_i])^2 \Var[\partial_{\ell_i} f(\ell_i,y_i)]}\right)^2\,, \label{eq:proof-thm1-use-lemma2}\\
    \leq & \left(\frac{1}{N}\sum_{i} \sqrt{2\Var[\partial_{\nu} \ell_i ]| \partial_{\ell_i}f(\ell_i,y_i)|^2_{max}+2(\mathbb{E}[\partial_\nu l_i])^2 | \partial_{\ell_i}f(\ell_i,y_i)|^2_{max}}\right)^2\,, \label{eq:proof-thm1-chain-inequalities}\\
    = &\left(\frac{1}{N}\sum_i g_i \sqrt{ \Var[\partial_\nu \ell_i]+ (\mathbb{E}[\partial_\nu \ell_i])^2} \right)^2\label{eq:proof-thm1-define-g}\,,
\end{align}
where:
\begin{itemize}
    \item In~\eqref{eq:var-d-generic-cost}, the inequality comes by applying Lemma~\ref{lem:var-sum}.
    \item In~\eqref{eq:proof-thm1-chain-rule}, we use the chain rule on the derivative of $f(\ell_i,y_i)$, that is, $\frac{\partial f(\ell_i,y_i)}{\partial \theta_\nu} = \frac{\partial f(\ell_i,y_i)}{\partial \ell_i}\cdot\frac{\partial \ell_i}{\partial \theta_\nu}$. Note that $y_i$ has no dependence on $\theta_\nu$.
    \item In~\eqref{eq:proof-thm1-use-lemma2},  we use Lemma~\eqref{lem:var-product}.
    \item In~\eqref{eq:proof-thm1-chain-inequalities}, we use the chain of inequalities $\Var[\partial_{\ell_i} f(\ell_i)] \le \mathbb{E}[(\partial_{\ell_i} f(\ell_i))^2] \leq | (\partial_{\ell_i} f(\ell_i))^2 |_{max} = | \partial_{\ell_i} f(\ell_i)|^2_{max}$.
    \item In~\eqref{eq:proof-thm1-define-g}, we define $g_i=\sqrt{2}| \partial_{\ell_i}f(\ell_i,y_i)|_{max}$. This completes the proof of the theorem.
\end{itemize}
\end{proof}

\section{Proof of Corollary~\ref{corollary:cor1}: Mean squared error and negative log-likelihood Loss Functions}\label{sec:appdx-cor}

\begin{corollary}[Barren plateaus in mean squared error and negative log-likelihood loss functions]\label{corollary:cor1-SM}
Consider the mean squared error loss function $\LC_{mse}$ defined in Eq.~\eqref{eq:loss-mse} and the negative log-likelihood loss $\LC_{log}$ defined in Eq.~\eqref{eq:loss-log}, with respective model-predicted labels $\tilde{y}_i(\thv)$ and model-predicted probabilities $p_i(y_i|\thv)$ both of the form of Eq.~\eqref{eq:measurement}. Suppose that the QNN has a barren plateau for the linear expectation values. That is, Eq.~\eqref{eq:var-BP} is satisfied for all $\tilde{y}_i(\thv)$ and $p_i(y_i|\thv)$. Then, assuming that $\tilde{y}_i(\thv) \in \OC(\poly(n))\; \forall i$ we have 
\begin{align}
    \Var[\partial_\nu \LC_{mse}] \in \mathcal{O}(1/\alpha^{n})\, , \label{eq:appdx-var-Lmse}
\end{align}
for  $\alpha > 1$. Similarly, assuming that $p_i(y_i|\thv) \in [b,1]\; \forall i,\thv$, where $b\in \Omega(1/\poly(n))$, we have
\begin{align}
    \Var[\partial_\nu \LC_{log}] \in \mathcal{O}(1/\alpha^{n})\, , \label{eq:appdx-var-Llog}
\end{align}
for  $\alpha > 1$.
\end{corollary}

\begin{proof}
First, we prove the result for the mean squared error loss function. Recall the mean squared loss function is defined as $\LC_{mse}(\vec\thv) = \frac{1}{N}\sum_{i=1}^{N}(y_i-\tilde{y}_i(\thv))^2$ (see Eq.~\eqref{eq:loss-mse}). Hence, in this case we have
\begin{align}
    \left|\frac{\partial f}{\partial \tilde{y}_i}\right|_{max}  
    & = 2(\tilde{y}_i^{max} - y_i)\,, 
\end{align}
where $\tilde{y}_i^{max}$ denotes the maximum value of $\tilde{y}_i$. Substituting this expression into Eq.~\eqref{eq:supp-var-generic-loss} we have
\begin{align}
    \Var[\partial_\nu \LC_{mse}(\theta)]& \leq  \frac{4}{N^2}\left( \sum_i (\tilde{y}_i^{max} - y_i) \sqrt{ \Var[\partial_\nu \tilde{y}_i]+ (\mathbb{E}[\partial_\nu \tilde{y}_i])^2} \right)^2.
\end{align}
Under the assumption that $\tilde{y}_i(\thv) \in \OC(\poly(n))\; \forall i,\thv$, and Eq.~\eqref{eq:var-BP} is satisfied for all $\tilde{y}_i(\thv)$ (i.e. the BP condition is met), we obtain Eq.~\eqref{eq:appdx-var-Lmse} as required.

Now we prove the result for the negative log-likelihood loss as in Eq.~\eqref{eq:loss-log}. We recall we can explicitly write the negative log-likelihood loss function as $\LC_{log}(\thv) = - \frac{1}{N} \sum_{i = 1}^{N} \log p_i(\thv)$, for which we have
\begin{align}
    \left|\frac{\partial f}{\partial p_i}\right|_{max} = \frac{1}{p_i^{min}} \, ,
\end{align}
where $p_i^{min}$ is the least possible value of $p_i(\thv)$. By substituting this into Eq.~\eqref{eq:supp-var-generic-loss} of Theorem~\ref{thm:var-generic-loss} and using the assumption $p_i(\thv) \in [b,1]\; \forall \vec{x}_i,\thv$, we have
\begin{align} \label{eq:var_log_general}
    {\rm Var}[\partial_\nu \LC_{log}(\thv)] &\leq \frac{2}{N^2}\left( \sum_i\left(\frac{1}{p_i^{min}}\right) \sqrt{ \Var[\partial_\nu p_i]+ (\mathbb{E}[\partial_\nu p_i])^2} \right)^2\,. 
\end{align}
Under the assumption that $p_i(\thv) \in [b,1]\; \forall i,\thv$ with $b\in \Omega(1/\poly(n))$, $N \in \OC(\poly(n))$, and Eq.~\eqref{eq:var-BP} is satisfied for all $p_i(\thv)$, we obtain Eq.~\eqref{eq:appdx-var-Llog} as required.
\end{proof}

\section{Extensions of Corollary~\ref{corollary:cor1} }\label{sec:appdx-cor-extensions}

We present two extensions to the supervised QML framework considered in Corollary \ref{corollary:cor1} where the true label  $y_i$ is a simple scalar. Specifically, we extend the results in Corollary~\ref{corollary:cor1} to the generalized mean square error loss function used in multivariate regression tasks and to loss functions based on the Kullback–Leibler divergence used in generative modelling. This shows how the results in our work are applicable to other loss functions used in machine learning tasks.  

\subsection{Generalized mean square error loss function for multivariate regression}
First, consider a multivariate regression task where the true label is now a vector of length $n_y$ i.e. $\vec{y}_i \in \RC^{n_y}$. For instance, here one can obtain the model-predicted label by measuring expectation values of $n_y$ different operators, each taking the form in Eq.~\eqref{eq:measurement} and representing a  distinct component of the vector. In this case, the generalized mean squared error reads
\begin{align}
    \LC^{(multi)}_{mse} = & \sum_i^N (\tilde{\vec{y}}_i(\thv) - \vec{y}_i)^2 \, ,\\
    = & \sum_i^N \sum_j^{n_y} (\tilde{y}^{(j)}_i(\thv) - y^{(j)}_i)^2 \, \label{eq:multi-regress-loss}, 
\end{align}
where $y_i^{(j)}$ ($\tilde{y}^{(j)}_i(\thv)$) is the $j$-th component of the $i$-th true (model-predicted) label. The form in Eq.~\eqref{eq:multi-regress-loss} is similar to the mean squared error loss function for a single valued true label in Eq.~\eqref{eq:loss-mse} with $n_y\times N$ training data. We now present a remark which generalizes the result for mean squared error in Corollary~\ref{corollary:cor1}.

\begin{remark}[Generalized mean squared error loss function]
Consider the generalized mean squared error loss function as defined in Eq.~\eqref{eq:multi-regress-loss} with the model-predicted labels $\tilde{\vec{y}}_i(\thv)$ where each component $\tilde{y}^{(j)}_i(\thv)$ is of the form defined in Eq.~\eqref{eq:measurement}. Suppose that the QNN has a barren plateau for linear expectation values, that is, Eq.~\eqref{eq:var-BP} is satisfied for all $\tilde{y}^{(j)}_i(\thv)$. Then, assuming that $\tilde{y}^{(j)}_i(\thv) \in \OC(\poly(n))\; \forall i,j$ we have 
\begin{align}
    \Var[\partial_\nu \LC^{(multi)}_{mse}]  \in \mathcal{O}(1/\alpha^{n})\, ,
\end{align}
for  $\alpha > 1$. 
\end{remark}
The proof proceeds with the same steps as in the consideration of the mean squared error loss in Corollary~\ref{corollary:cor1}. 

\subsection{KL-divergence-based loss function for generative modelling}
Now, consider generative modelling which is an unsupervised learning task (i.e. no given true labels). The goal is to learn an (unknown) underlying probability distribution $Q'(\vec{x})$ which generates a training sample set $\{ \vec{x}_i\}_{i=1}^N$. Unlike supervised learning tasks, here the trained model is expected to be capable of providing new samples with the optimized underlying probability distribution $P(\vec{x}; \thv)$, which are as close to $Q'(\vec{x})$ as possible. 

During the training process, it is common to choose loss functions as the KL divergence and reverse KL divergence, which take the respective forms 
\begin{align}
    \LC_{KL} &  = - \sum_{\vec{x}\in X} Q(\vec{x}) \log\left( \frac{P(\vec{x};\thv)}{Q(\vec{x})}\right) \,, \label{eq:loss-kld}\\
    \LC_{rev-KL} & = -  \sum_{\vec{x}\in X} P(\vec{x};\thv) \log\left( \frac{Q(\vec{x})}{P(\vec{x};\thv)}\right) \,. \label{eq:rev-loss-kld}
\end{align}
Here $Q(\vec{x}) = \sum_{i;\vec{x}_i = \vec{x}} \vec{x}_i/N$ is the probability of obtaining a sample $\vec{x}$ constructed from the training dataset, and $P(\vec{x};\thv)$ is the estimated probability of obtaining a sample $\vec{x}$ from the model. 
As $Q(\vec{x})$ is independent of the variational parameters $\thv$, $\LC_{KL}$ and $\LC_{rev-KL}$ can be treated in a similar way as $\LC_{log}$ in Eq.~\eqref{eq:loss-log}, leading to the following remark.

\begin{remark}[KL-divergence-based loss functions]
Consider the KL divergence and reverse KL divergence loss functions as defined in Eq.~\eqref{eq:loss-kld} and Eq.~\eqref{eq:rev-loss-kld} with the underlying model probability distribution 
$P(\vec{x};\thv)$ taking the form in Eq.~\eqref{eq:measurement} for all $\vec{x} \in X$. Suppose that the QNN has a barren plateau for the linear expectation values. That is, Eq.~\eqref{eq:var-BP} is satisfied for all $P(\vec{x};\thv)$. Then, assuming that $P(\vec{x};\thv),Q(\vec{x}) \in [b,1] \; \forall \vec{x},\thv$, where $b\in 1/\Omega(\poly(n))$, we have 
\begin{align}
    \Var[\partial_\nu \LC_{KL}] & \in \mathcal{O}(1/\alpha^{n})\, , 
\end{align}
and
\begin{align}
    \Var[\partial_\nu \LC_{rev-KL}] & \in \mathcal{O}(1/\alpha^{n})\, , \label{eq:appx-bp-rev-kld}
\end{align}
for  $\alpha > 1$.
\end{remark}

\begin{proof}
The proof for the result for KL divergence proceeds with the same steps as the proof for the negative log-likelihood in Corollary~\ref{corollary:cor1}.

For the result on the reverse KL divergence, we note that the loss function is a sum of functions $f(P(\vec{x},\thv), \vec{x}) = - P(\vec{x};\thv) \log\left( \frac{Q(\vec{x})}{P(\vec{x};\thv)}\right)$, leading to
\begin{align}
    \left| \frac{\partial f}{\partial P(\vec{x})}  \right|_{max} = \left| 1 + \log\left( \frac{P(\vec{x})}{Q(\vec{x})}\right)\right|_{max}\, .
\end{align}
By substituting this into Eq.~\eqref{eq:supp-var-generic-loss} of Theorem~\ref{thm:var-generic-loss}, we have
\begin{align} \label{eq:var_rev_kld}
    \Var[\partial_\nu \LC_{rev-KL}] &\leq \frac{2}{|X|^2} \sum_{\vec{x}\in X}\left(\left| 1 + \log\left(\frac{P(\vec{x})}{Q(\vec{x})}\right)\right|_{max} \sqrt{ \Var[\partial_\nu P(\vec{x};\thv)]+ (\mathbb{E}[\partial_\nu P(\vec{x};\thv)])^2} \right)^2\, , 
\end{align}
where $|X|$ is the cardinality of $X$. 
Under the assumption $Q(\vec{x}), P(\vec{x};\thv) \in [b,1] \; \forall \vec{x},\thv$, where $b\in 1/\Omega(\poly(n))$ and Eq.~\eqref{eq:var-BP} is satisfied $P(\vec{x};\thv)\;\forall \vec{x}$, we obtain Eq.~\eqref{eq:appx-bp-rev-kld} as required.
\end{proof}
We remark that similar to the negative log-likelihood, the assumption on $Q(\vec{x})$ and $P(\vec{x};\thv)$ to clip their values to be strictly greater than zero is common in practice when using loss functions based on the KL-divergence~\cite{abadi2016tensorflow}.

\section{Fisher information matrix results}\label{sec:appdx-FI}
In this section we show that under the conditions for which the logarithmic loss function shows a BP, the matrix elements of the FI matrix probabilistically exponentially vanish.

\begin{proposition}\label{prop:FI-SM}
Under the assumptions of Corollary \ref{corollary:cor1} for which the negative log-likelihood loss function has a BP according to Eq.~\eqref{eq:var-BP}, and assuming that the number of trainable  parameters in the QNN is in $ \OC(\poly(n))$, we have
\begin{equation}
    \mathbb{E}\big[\big|\tilde{F}_{\mu\nu}(\thv)\big|\big] \leq G(n) \,, \; \textrm{with}\;\, G(n) \in \OC(1/\alpha^n)\,,
\end{equation}
where $\tilde{F}_{\mu\nu}(\thv)$ are the matrix entries of $\tilde{F}(\thv)$ as defined in Eq.~\eqref{eq:FI_empirical}.
\end{proposition}

\begin{proof}
From the proof of Theorem \ref{thm:var-generic-loss} we have 
\begin{align}
    \Var[\partial_\nu f(\ell_i,y_i)] \leq g_i^2 \left( \Var[\partial_\nu \ell_i]+ (\mathbb{E}[\partial_\nu \ell_i])^2 \right) \label{eq:appdx-var-partial-f}
\end{align}
for all $i$. 
We note that the expectation value is similarly bounded
\begin{align}
    \mathbb{E}[\partial_\nu f(\ell_i,y_i)] &= \mathbb{E}[\partial_{\ell_i} f(\ell_i,y_i) \partial_{\nu} \ell_i] \\
    &=  \mathbb{E}[\partial_{\ell_i}f(\ell_i,y_i)] \mathbb{E}[\partial_{\nu} \ell_i] + \Cov[\partial_{\ell_i}f(\ell_i,y_i), \partial_{\nu} \ell_i]\\
    &\leq \sqrt{\Var[\partial_{\ell_i}f(\ell_i,y_i)]\Var[\partial_{\nu} \ell_i]}\\
    &= \frac{1}{\sqrt{2}} g_i \sqrt{\Var[\partial_{\nu} \ell_i]}\,. \label{eq:appdx-E-partial-f}
\end{align}
where the first equality is an application of the chain rule, the second equality is due to the definition of the covariance, the first inequality is due to the Cauchy-Schwarz inequality, and in the final equality we have used the definition of $g_i$ in Eq.~\eqref{eq:g_i-SM}.
This enables us to write the bound
\begin{align}
    \mathbb{E}\left[\partial_\nu f(\ell_i,y_i)\partial_\mu f(\ell_i,y_i)\right] &\leq \mathbb{E}\left[\partial_\nu f(\ell_i,y_i)]\mathbb{E}[\partial_\mu f(\ell_i,y_i)\right] + \Cov\left[\partial_\nu f(\ell_i,y_i),\partial_\mu f(\ell_i,y_i)\right] \\
    &\leq \mathbb{E}\left[\partial_\nu f(\ell_i,y_i)]\mathbb{E}[\partial_\mu f(\ell_i,y_i)\right] + \sqrt{\Var\left[\partial_\nu f(\ell_i,y_i)\right]\Var\left[\partial_\mu f(\ell_i,y_i)\right]}\\
    &\leq \frac{1}{2}g_i^2\sqrt{\Var[\partial_{\nu} \ell_i]\Var[\partial_{\mu} \ell_i]} + g_i^2\sqrt{\left( \Var[\partial_\nu \ell_i]+ (\mathbb{E}[\partial_\nu \ell_i])^2 \right)\left( \Var[\partial_\mu \ell_i]+ (\mathbb{E}[\partial_\mu \ell_i])^2 \right)}
\end{align}
where again, the first two inequalities come from the definition of the covariance and an application of the Cauchy-Schwarz inequality, and in order to obtain the third inequality we have used Eq.~\eqref{eq:appdx-var-partial-f} and Eq.~\eqref{eq:appdx-E-partial-f}. We now consider the negative log-likelihood loss function, that is where $\ell_i = p_i$ are probabilities and $f(\ell_i,y_i)=\log p_i$ for all $i$.
Then, if our assumptions are satisfied, namely Eq.~\eqref{eq:var-BP} is satisfied for all $p_i(\thv)$, and $p_i(\thv) \in [b,1]\; \forall i,\thv$ with $b\in \Omega(1/\poly(n))$, we have
\begin{align}
    \mathbb{E}[\tilde{F}_{\mu\nu}]= \mathbb{E}\left[\sum^{N}_i\partial_\mu \log p_i(\thv)\partial_\nu \log p_i(\thv)\right] \leq H(n)\,, \; \textrm{with}\; H(n) \in \mathcal{O}(1/\alpha^{n}) \label{eq:appdx-E-F}
\end{align}
for all $\theta_\mu$ and $\theta_\nu$. By considering diagonal terms, this means that we can bound $\Tr[\mathbb{E}[\tilde{F}_{\mu\nu}]]\leq N_pH(n)$, where $N_p$ is the dimension of parameter space. As $\tilde{F}$ is positive semidefinite, the individual eigenvalues of $\tilde{F}$ also each satisfy this bound, and this allows us to bound the individual matrix elements as
\begin{align}
    \mathbb{E}\big[\big|\tilde{F}_{\mu\nu}\big|\big] \leq G(n) \,, \; \textrm{with}\; G(n) = N_pH(n)\,,
\end{align}
as required.
\end{proof}

Now we show a corollary that assumes the QNN structure allows us to use the parameter shift rule to evaluate partial derivatives.

\begin{corollary}
Under the assumptions of Corollary \ref{corollary:cor1} for which the negative log-likelihood loss function has a BP according to Eq.~\eqref{eq:var-BP}, and assuming that the QNN structure allows for the use of the parameter shift rule, we have
\begin{equation}
    \emph{Pr}\left(\left|\tilde{F}_{\mu\nu} - \mathbb{E}[\tilde{F}_{\mu\nu}] \right| \geq c\right) \leq \frac{Q(n)}{c^2}\,,
\end{equation}
where $\mathbb{E}[\tilde{F}_{\mu\nu}]\leq H(n)$ and $Q(n),H(n)\in \OC(1/\alpha^n)$.
\end{corollary}

\begin{proof}
We have
\begin{align}
    \Var\left[\sum_i^{N}\partial_\mu f(\ell_i,y_i)\partial_\nu f(\ell_i,y_i)\right]&\leq \left( \sum^{N}_i \sqrt{\Var\left[ \partial_\mu f(\ell_i,y_i) \partial_\nu f(\ell_i,y_i)\right]} \right)^2 \\
    &\leq \left( \sum^{N}_i \sqrt{2\Var[\partial_\mu f(\ell_i,y_i)] |(\partial_\nu f(\ell_i,y_i))^2|_{max} + 2(\mathbb{E}[\partial_\nu f(\ell_i,y_i)])^2\Var[\partial_\nu f(\ell_i,y_i)]}\right)^2 \\
    &\leq \left( \sum^{N}_i \sqrt{2} \left|(\partial_\nu f(\ell_i,y_i))\right|_{max} \sqrt{\Var[\partial_\mu f(\ell_i,y_i)] + 2(\mathbb{E}[\partial_\mu f(\ell_i,y_i)])^2}\right)^2\\
    &= \left( \sum^{N}_i \sqrt{2} g_i \left|(\partial_\nu \ell_i)\right|_{max} \sqrt{\Var[\partial_\mu f(\ell_i,y_i)] + 2(\mathbb{E}[\partial_\mu f(\ell_i,y_i)])^2}\right)^2\\
    &\leq \left( \sum^{N}_i \sqrt{2} g_i^2 \left|(\partial_\nu \ell_i)\right|_{max} \sqrt{  2\Var[\partial_\nu \ell_i]+ (\mathbb{E}[\partial_\nu \ell_i])^2 }\right)^2 \label{eq:appdx-var-partial-f-2}
\end{align}
where in the first inequality we have used Supplementary Lemma \ref{lem:var-sum}, in the second inequality we have used Supplementary Lemma \ref{lem:var-product}, the third inequality comes by bounding $\Var[\partial_\nu f(\ell_i,y_i)]$ by the maximum value of $f(\ell_i,y_i)^2$, the penultimate line is an application of the chain rule, and the final inequality comes by using Eq.~\eqref{eq:appdx-var-partial-f} and Eq.~\eqref{eq:appdx-E-partial-f} from Proposition \ref{prop:FI-SM}.
Now consider the logarithmic loss function, where $l_i(\thv) = p_i(\thv)$ is a probability. Assuming that the QNN structure allows for the use of the parameter shift rule, we have \begin{align}
    |(\partial_\nu p_i)|_{max} = \frac{1}{2}|(p_i(\thv_+)  - p_i(\thv_-))|_{max}\leq \frac{1}{2}\,,\label{eq:appdx-partial-p}
\end{align}
for all $i$, where we have used the fact that $p_i(\thv_+)$ and $p_i(\thv_-)$ are probabilities and thus bounded. This allows us to write,
\begin{align}
    \Var[\tilde{F}_{\mu\nu}] &= \Var\left[\sum_i^{N}\partial_\mu \log p_i\,\partial_\nu \log p_i\right]\\
    &\leq \frac{1}{2}\left( \sum^{N}_i g_i^2 \sqrt{  2\Var[\partial_\nu p_i]+ (\mathbb{E}[\partial_\nu p_i])^2 }\right)^2\,,
\end{align}
where in the first line we have explicitly written the matrix element of the empirical FI matrix, and the inequality comes from Eq.~\eqref{eq:appdx-var-partial-f-2} and substituting in the bound Eq.~\eqref{eq:appdx-partial-p}. Then, if our assumptions are satisfied, namely Eq.~\eqref{eq:var-BP} is satisfied for all $p_i(\thv)$, and $p_i(\thv) \in [b,1]\; \forall i,\thv$ with $b\in \Omega(1/\poly(n))$, we have
\begin{align}
    \Var[\tilde{F}_{\mu\nu}] &\leq Q(n)\,, \; \textrm{with}\; Q(n) \in \mathcal{O}(1/\alpha^{n}) \label{eq:appdx-var-F}\,,
\end{align}
for all $\theta_\mu$ and $\theta_\nu$. By inspecting Eq.~\eqref{eq:appdx-E-F} and Eq.~\eqref{eq:appdx-var-F} and using Chebyshev's inequality we have 
\begin{equation}
    \textrm{Pr}\left(\left|\tilde{F}_{\mu\nu} - \mathbb{E}[\tilde{F}_{\mu\nu}] \right| \geq c\right) \leq \frac{\Var[\tilde{F}_{\mu\nu}]}{c^2} \leq \frac{Q(n)}{c^2}\,,
\end{equation}
where $\mathbb{E}[\tilde{F}_{\mu\nu}]\leq H(n)$, where $Q(n),H(n)\in \OC(1/\alpha^n)$.
\end{proof}

\subsection{Assuming log-likelihood loss has BP is insufficient to show exponentially small FI matrix elements}

To prove the above results, we make the assumption that linear expectation values have exponentially vanishing variance of their partial derivatives (i.e. they display a barren plateau). Here we show that relaxing these assumptions, such that the negative log-likelihood loss function has a plateau, is insufficient to guarantee exponentially vanishing elements of the FI matrix. We remark that this is an artefact of the fact that in QML, the loss functions are constructed from many data points. Consider the case where the dataset consists of 2 data points $\{ (\vec{x}_1,y_1), (\vec{x}_2,y_2) \}$. The log-likelihood in this case reads
\begin{align}
    \partial_\nu \LC_{log}(\thv) = \frac{1}{2}\left( \partial_\nu \log(p(y_1|\vec{x}_1;\thv) + \partial_\nu\log(p(y_2|\vec{x}_2;\thv) \right) \,.
\end{align}
In the presence of barren plateaus in the log-likelihood loss landscape, we have
\begin{align}
    \Var[\partial_\nu \LC_{log}(\thv)] & =  \mathbb{E}[(\partial_\nu \LC_{log})^2] -  (\mathbb{E}[\partial_\nu \LC_{log}])^2 \, , \\
    & = \mathbb{E}\left[\frac{1}{4}\left( \partial_\nu \log(p(y_1|\vec{x}_1;\thv) + \partial_\nu\log(p(y_2|\vec{x}_2;\thv) \right)^2\right] \,, \label{eq:appdx-E-2datapoints} \, ,
\end{align}
where in the second line we have used the fact that  $\mathbb{E}[\partial_\nu \LC_{log}]=0$~\cite{cerezo2020cost,sharma2020trainability,pesah2020absence}. From Eq.~\eqref{eq:appdx-E-2datapoints} it is clear to see that the barren plateau condition ($\Var[\partial_\nu \LC_{log}(\thv)]$ being exponentially vanishing)  can be satisfied with $\{\log(p(y_1|\vec{x}_i;\thv)\}_{i \in \{1,2\}}$ having exponentially similar magnitudes but opposite signs across the landscape. Hence, in order to guarantee exponentially small FI matrix elements, it is not sufficient to only assume that the negative log-likelihood loss function has a barren plateau. Rather, extra assumptions have to be made, for instance that the linear expectation values have a barren plateau, as assumed in Corollary \ref{corollary:cor1}. 

\section{Numerical implementations}\label{sec:appdx-numerics}
In this section we provide technical details for our numerical results. 

\subsection{Dimensional reduction of features on the MNIST dataset}

Here we describe explicitly how the images in the MNIST dataset are  reduced to length-$n$ real-valued vectors using principal component analysis (PCA)~\cite{jolliffe2005principal}. Consider a dataset $\{\vec{x}_i, y_i\}_i^{N}$ with $N$ data points, such that each input $\vec{x}_i$ is a vector of length $784$ (obtained by vectorizing the gray-scaled $28\times 28$ image). We perform PCA on the dataset with the following steps.
\begin{enumerate}
    \item Compute the average of the input data points $\vec{x}_{avg}= \sum_i \vec{x}_i / N$.
    \item Normalize each individual input data point as $\tilde{\vec{x}}_i = \vec{x}_i - \vec{x}_{avg}$ for all $i$.
    \item Construct a matrix of this normalized dataset $X = (\tilde{\vec{x}}_1, \tilde{\vec{x}}_2, ... ,\tilde{\vec{x}}_{N})^T$ of size $N \times 784$ where the $i-$th row of the matrix $X$ represents a $i-$th input data point, $\vec{x}_i^T$.
    \item From the matrix $X$, construct a square matrix $X^TX = \sum_i \tilde{\vec{x}_i}^T\tilde{\vec{x}_i} $ of size $784 \times 784$ and perform an eigenvalue decomposition of $X^T X$.
    \item Keep the $n$ eigenvectors corresponding to the largest $n$ eigenvalues of $X^TX$. These are the principal components of the dataset. Construct a matrix $M$ of size $784 \times n$ with $n$ column vectors corresponding to these $n$ eigenvectors.
    \item Compute $X M$, leading to a new data matrix of size $N \times n$. This matrix multiplication corresponds to a projection of each individual input data point into the sub-space formed by these $n$ principal components. 
\end{enumerate}

\subsection{CHE architecture}\label{appx:che-architecture}
Here, we describe in detail the architecture of the CHE in Fig.~\ref{fig:embedding}(c). The embedding, originally proposed in~\cite{havlivcek2019supervised}, is based on the Instantaneous Quantum Polynomial (IQP) architecture which takes the form of $U_{IQP} = H^{\otimes n} U_\mathcal{Z}  H^{\otimes n}$ where $H$ is the Hadamard gate and $U_\mathcal{Z}$ is an arbitrary random diagonal unitary acting on all qubits. Sampling from the output distribution of this IQP circuit has been analytically shown to be classically hard to simulate and proposed as one of the early approaches for demonstrating quantum supremacy~\cite{bremner2011classical}. In a similar fashion, the CHE scheme has been conjectured to be classically hard to simulate for more than two layers. As shown in Fig.~\ref{fig:che-architecture},  one layer consists of the Hadamard gates on all qubits followed by the data-encoded unitary $W(\vec{x}_i)$ which consists of single and two-qubit gates diagonal in the computational basis. 
Specifically, given that $\vec{x}_i$ has length $n$, $W(\vec{x}_i)$ is defined as
\begin{align}
    W(\vec{x}_i) =  \left( \prod_{j<k} e^{-i x^{(j)}_i x^{(k)}_i Z_jZ_k} \right)\left(\prod_{j=1}^n e^{-ix^{(j)}_i Z_j} \right)\, ,
\end{align}
where $x_i^{(j)}$ is the $j$-th component of $\vec{x}_i$, and $Z_j$ is the Pauli-$Z$ operator on $j$-th qubit. We note that
$e^{-ix^{(j)}_i Z_j}$ is a single-qubit rotation on $j-$th qubit encoding $j$-th component of the $\vec{x}_i$, and $e^{-i x^{(j)}_i x^{(k)}_i Z_jZ_k}$ is a two-qubit $ZZ$ gate that encodes the product of the $j$-th and $k$-th components of the data. The two-qubit $ZZ$ gates act on all possible pairs of qubits. 
\begin{figure}
	\includegraphics[width= .8 \columnwidth]{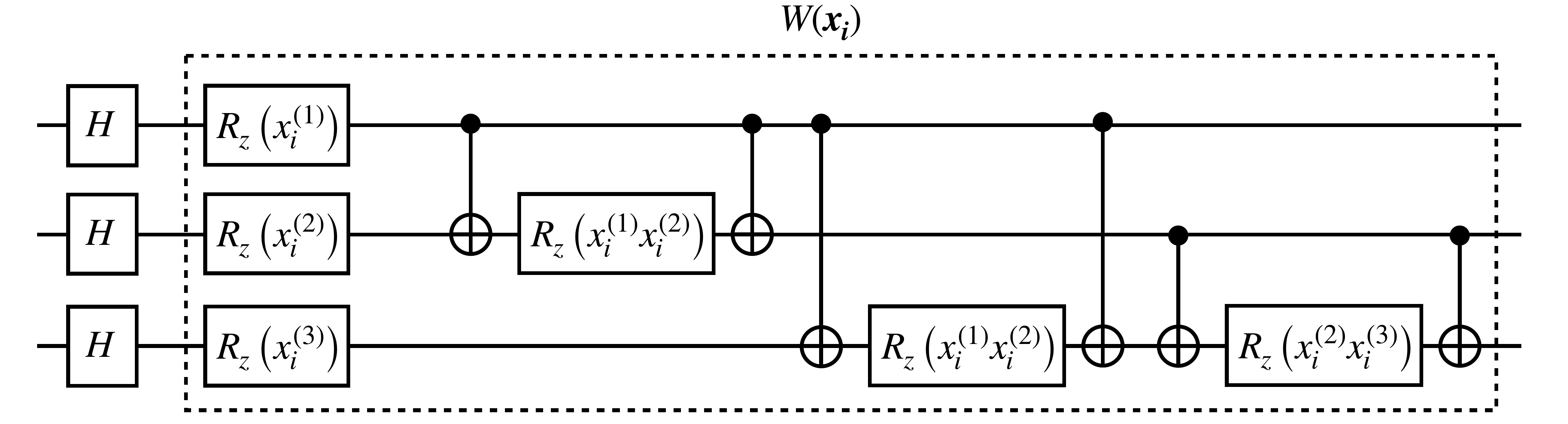}
	\caption{\textbf{CHE architecture.} One layer of the embedding is comprised of an $R_Z$ rotation on each qubit, where the rotation angle on the $j$-th qubit corresponds to the $j$-th component of $\vec{x}_i$, followed by a series of two-qubit $ZZ$ gates on all pairs of qubits. For each pair of qubits $j$ and $k$, the two-qubit ZZ gate encodes the product of the $j$-th and $k$-th components of $\vec{x}_i$.
	}\label{fig:che-architecture}
\end{figure}

\subsection{QCNN architecture}\label{appx:qcnn-architecture}
We now provide details on the QCNN architecture used in this work. QCNNs are motivated by the structure of classical convolutional neural networks (which in turn are motivated by the structure of the visual cortex). Thus, in a QCNN, a series of convolutional layers are interleaved with pooling layers which  reduce the number of degrees of freedom, while preserving the relevant features of the input state~\cite{cong2019quantum}. Effectively,  after each pooling layer,  the number of remaining qubits is reduced by (about) half. Thus, for an initial input of $n$ qubits, the total depth of the QCNN is $\OC(\log(n))$. As illustrated in Fig.~\ref{fig:qcnn}, the convolutional layer is comprised of two layers of parametrized two-qubit unitary blocks acting on alternating pairs of nearest-neighbor qubits. In the pooling layer, we apply CNOT gates on pairs of nearest-neighbor qubits where the controlled qubits are discarded for the later parts of the circuit. We note that, although one could consider more general measurements and controlled unitaries in the pooling layers in a more general setting, we do not consider such measurements in our numerics. Finally, after reducing down to two qubits, another parametrized unitary block (as in the convolutional layer) together with parametrized single-qubit $R_X$ rotation gates are applied before the measurement.

\begin{figure}
	\includegraphics[width= .8 \columnwidth]{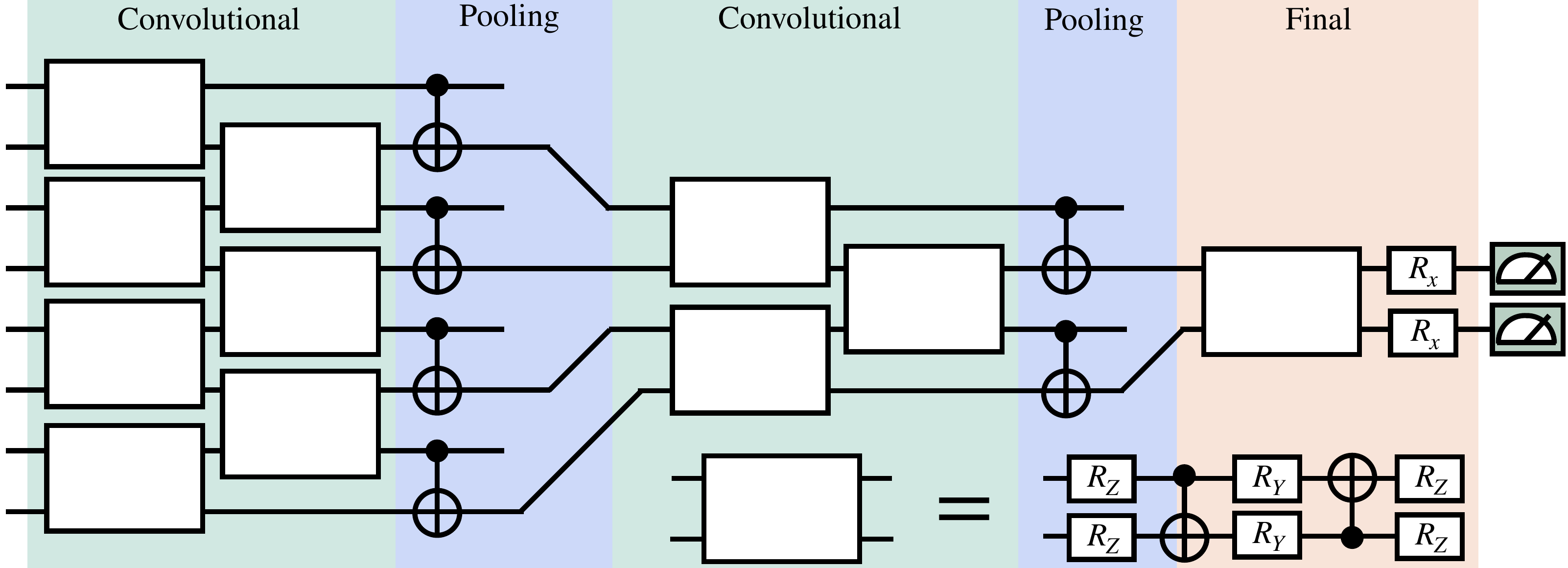}
	\caption{\textbf{QCNN architecture.} We illustrate how the QCNN circuit is constructed for 8 qubits. After all convolutional and pooling layers reducing the system size is reduced down to 2 qubits. Then, one applies a final layer comprised of a parametrized 2-qubit unitary block and parametrized single qubit $R_X$ rotations before performing measurements. We also present the structure of the parametrized 2-qubit unitary blocks utilized in this work .
	}\label{fig:qcnn}
\end{figure}

\end{document}